\documentclass[english,aip,jmp,amsmath,amssymb,reprint,nofootinbib]{revtex4}
\usepackage[T1]{fontenc}
\usepackage[latin9]{inputenc}
\usepackage{array}
\usepackage{float}
\usepackage{mathtools}
\usepackage{amsmath}
\usepackage{amsthm}
\usepackage{amssymb}
\usepackage{esint}

\makeatletter

\providecommand{\tabularnewline}{\\}

\@ifundefined{textcolor}{}
{%
 \definecolor{BLACK}{gray}{0}
 \definecolor{WHITE}{gray}{1}
 \definecolor{RED}{rgb}{1,0,0}
 \definecolor{GREEN}{rgb}{0,1,0}
 \definecolor{BLUE}{rgb}{0,0,1}
 \definecolor{CYAN}{cmyk}{1,0,0,0}
 \definecolor{MAGENTA}{cmyk}{0,1,0,0}
 \definecolor{YELLOW}{cmyk}{0,0,1,0}
}
  \theoremstyle{definition}
  \newtheorem{defn}{\protect\definitionname}
  \theoremstyle{plain}
  \newtheorem{lem}{\protect\lemmaname}
  \theoremstyle{plain}
  \newtheorem{cor}{\protect\corollaryname}
\theoremstyle{plain}
\newtheorem{thm}{\protect\theoremname}

\usepackage{dcolumn}
\usepackage{bm}
\usepackage{yfonts}
\usepackage{dsfont}
\usepackage{amsxtra}
\usepackage[english]{babel}
\usepackage[arrow, matrix, curve]{xy}
\usepackage{tikz}

  \providecommand{\definitionname}{Definition}
  \providecommand{\lemmaname}{Lemma}
\providecommand{\corollaryname}{Corollary}
\providecommand{\theoremname}{Theorem}

\usepackage{babel}
\providecommand{\definitionname}{Definition}
  \providecommand{\lemmaname}{Lemma}
\providecommand{\corollaryname}{Corollary}
\providecommand{\theoremname}{Theorem}

\makeatother

\usepackage{babel}
  \providecommand{\definitionname}{Definition}
  \providecommand{\lemmaname}{Lemma}
\providecommand{\corollaryname}{Corollary}
\providecommand{\theoremname}{Theorem}

\begin{document}
\global\long\def\bm{M}

\global\long\def\BM{\mathcal{M}}

\global\long\def\vb{E}

\global\long\def\VB{\mathcal{E}}

\global\long\def\fb{V}

\global\long\def\FB{\mathcal{V}}

\global\long\def\cd#1{D_{#1}}

\global\long\def\CD#1{\mathcal{D}_{#1}}

\global\long\def\jb{Jb\left(E\right)}

\global\long\def\js{J}

\global\long\def\JB{\mathcal{J}b\left(\mathcal{E}\right)}

\global\long\def\JS{\mathcal{J}}

\global\long\def\dom#1{\mathfrak{D}^{#1}}

\global\long\def\d#1{\text{d}_{#1}}

\global\long\def\vol{\text{vol}}

\global\long\def\VOL{\text{Vol}}

\global\long\def\mult{\text{\ensuremath{\mathfrak{M}}}}

\global\long\def\ef{\text{\ensuremath{\vec{\phi}}}}

\title{Generalised conservation laws in non-local field theories}

\author{\textbf{Alexander Kegeles$^{1,2}$} and \textbf{Daniele Oriti$^{1}$}}

\address{\vspace{0.5cm}
 $^{1}$ Max Planck Institute for Gravitational Physics (Albert Einstein
Institute), \\
 Am Mühlenberg 1, 14476 Potsdam-Golm, Germany, EU \\
 alexander.kegeles@aei.mpg.de, daniele.oriti@aei.mpg.de \\
 \\
 $^{2}$University of Potsdam Institute of Physics and Astronomy\\
 Karl-Liebknecht-Strasse 24/25, 14476 Potsdam-Golm, Germany, EU\\
 \\
 }
\begin{abstract}
We propose a geometrical treatment of symmetries in non-local field
theories, where the non-locality is due to a lack of identification
of field arguments in the action. We show that the existence of a
symmetry of the action leads to a generalised conservation law, in
which the usual conserved current acquires an additional non-local
correction term, obtaining a generalisation of the standard Noether
theorem. We illustrate the general formalism by discussing the specific
physical example of complex scalar field theory of the type describing
the hydrodynamic approximation of Bose-Einstein condensates. We expect
our analysis and results to be of particular interest for the group
field theory formulation of quantum gravity. 
\end{abstract}
\maketitle

\section{Introduction}

The notion of \textit{symmetry} is at the very foundation of our description
of physical systems, in particular classical and quantum field theories.
Symmetries characterise their physical content and interpretation.
Their particle content is defined in terms of representations of the
Poincaré group, and their classical and quantum dynamics is characterised
e.g. by restrictions on the allowed field interactions. Symmetries
are also one of the main tools to characterise their macroscopic properties,
since different macroscopic phases of a given physical system are
often efficiently characterised by symmetries and their (spontaneous)
breaking. In fact, conservation laws following from symmetries may
encode basically the whole dynamical content of many macroscopic systems,
hydrodynamics being the obvious example. The corresponding universality
class is largely independent of the microscopic details of the system,
while the corresponding conserved charges identify key physical quantities.
In local (quantum) field theories the connection between symmetries
of the action and physically measurable quantities like conserved
currents and charges is very well understood. This relation is given
by the Noether theorem \cite{Noether:1971aa}, that allows to identify
and compute the conserved quantities and corresponding symmetry transformations.

The other main ingredient in our description of physical systems is
indeed the notion of locality, which again largely dictates the structure
of fundamental interactions and it is, both historically and conceptually,
at the root of the very notion of fields.

Still, the necessity to go beyond the framework of local field theories
has been voiced in several contexts. Non-local field theories are
routinely used in condensed matter theory and many-body quantum theory,
in the hydrodynamic approximation, e.g. in the theory of Bose-Einstein
condensates \cite{Leggett1,Leggett2} and even in relativistic theories
\cite{Kristensen:1952aa,Pauli:1953aa}. In this context, the field
theory fails to be local in that the typical interaction term depends
on more than a single spacetime point, i.e. the fields interact at
a distance. Formally, this implies that the Lagrangian of the theory
has a domain given by a direct product of several copies of the (spacetime)
manifold, on which fields are defined. However, the consequences of
this non-local structure are rarely investigated in detail, mainly
because such field theories are understood as the non-relativistic
approximation of truly local field theories, and thus their non-local
features are not considered of any fundamental significance.

More radical in scope are the proposals\footnote{Notice that most such models are non-local in the sense that they
involve an infinite number of field derivatives (see for example \cite{nonlocal2,nonlocal3,nonlocal4,nonlocal5,nonlocal6}).
It is unclear to us if there is a general connection between this
type of non-locality and the one we deal with in our analysis. At
any rate, we will not deal with this type of field theories.} for non-local field theory descriptions of cosmological phenomena
(see for example \cite{nonlocCosmo,nonlocCosmo2,nonlocCosmo3}) and
of semi-classical black hole physics (see for example \cite{nonlocBH,nonlocBH2,nonlocBH3}).
In all these cases the breakdown of strict locality, even in spacetime-based
physics, is not understood as an artifact of the approximation used,
but rather as the signal of an even more drastic departure from the
local description of fundamental interactions at a deeper level of
reality \cite{sorkin}. And locality is indeed the main feature of
standard spacetime physics that is affected by effective models of
quantum gravity, which suggest a basis for quantum gravity phenomenology.
Deformations of spacetime symmetries \cite{dsr,dsr2} and/or ideas
from non-commutative geometry \cite{NCG,NCG2} are just few examples
with fundamental non localities. In fact, several hints are accumulating
that a more fundamental description of quantum gravity will imply
a breakdown of the very notion of spacetime and its \lq dissolution\rq
~ at microscopic scales, only to emerge in a suitable continuum and
semiclassical approximation from the collective dynamics of more fundamental,
non-spatiotemporal degrees of freedom (see the discussion in \cite{emergence,emergence2,emergence3,emergence4},
and references therein). The more fundamental quantum gravity description,
in this perspective, will not be spacetime-based, by definition, and
thus will not be \lq local\rq~ in any usual sense. In itself this
does not imply that a field theory framework cannot be devised for
describing the fundamental degrees of freedom, even though it requires
a more abstract type of field theories which obviously are not defined
on a spacetime manifold. This more general class of field theories
have been developed and go under the name of \lq group field theories\rq~
\cite{GFT,GFT2,GFT3,GFT4,GFT5}.

Group field theories are quantum field theories defined on a group
manifold, not to be identified with spacetime, but rather defining
the dynamics of \lq quanta \textit{of} spacetime\rq. This formalism
lies at the point of convergence of several quantum gravity approaches,
being understood both as a second quantised reformulation of loop
quantum gravity \cite{GFT-LQG,GFT-LQG2,GFT-LQG3,LQG,LQG2} and as
a generalisation and enrichment of tensor models \cite{tensor,tensor2}
and simplicial gravity. Not being defined on spacetime, the standard
type of locality has no reason to be invoked for these field theories.
Indeed group field theories are non-local with respect to the group
manifold that represents the domain of the fundamental fields. In
fact, a peculiar non-local pairing of field arguments entering their
interactions is a key characterising feature of the formalism. Moreover,
a similar type of non-locality is present in a different type of field
theories which lead to non-linear extensions of quantum cosmology
\cite{martin}, where the wave function interacts with itself non-locally
on minisuperspace. The similarity is not entirely mysterious, given
that the latter type of effective cosmological dynamics emerges from
group field theories as the hydrodynamic approximation of the quantum
dynamics of condensate states \cite{GFTcosmo,GFTcosmo2,GFTcosmo3,GFTcosmo4,GFTcosmo5,GFTcosmo6,GFTcosmo7},
and in a way which is completely analogous to what happens in real
Bose-Einstein condensates. This type of quantum gravity models are
our main motivation for considering the issue of symmetries and conservation
laws in non-local field theories.

The issue is wide open. Indeed, in non-local field theories the standard
Noether theorem fails and a general relation between symmetries and
conservation laws is not known. For flat base manifolds and a special
kind of non-local systems the connection between symmetries and conservation
laws can be established \cite{Pauli:1953aa,Bloch:1950aa,Marnelius:1973aa,Huang:2012aa,Yukawa:1950aa}.
But a general treatment of the problem is lacking. And in the context
of group field theories, a first, partial analysis was presented in
\cite{josephNoether}.

Our goal, in this paper, is to provide a solid and general analysis
of the connection between symmetries and conservation laws in non-local
field theories, where the non-locality amounts to a lack of simple
identification of field arguments in the Lagrangian.

We use differential geometrical language, in which the Lagrangian
is seen as a function on many copies of a single first jet bundle,
but the geometry of the base manifold of a single bundle is left unspecified
\cite{Abraham:2008aa}. This allows for a very general approach to
the problem and leads to a result that can be applied to any non-local
theory on curved space time in which the Lagrangian depends on first
derivatives of the fields at most. This includes a large class of
field theories used in condensed matter systems, and our motivating
quantum gravity models, that is group field theories.

We derive the connection between symmetries and conserved quantities
for non-local theories and obtain a Noether-like conservation law
augmented by a \lq correction term\rq~ that results from the presence
of non-locality. The resulting continuity equation is our main result
and is what we call the \lq generalised conservation law\rq. The
result is rigorous and very general, for the class of theories it
applies to but also because it concerns any continuous symmetry. Of
course its main interest lies in the application to specific physical
examples. In order to clarify and illustrate in detail our formalism,
we present an explicit application of our result to the case of a
field theory with 2-body interaction, as the one describing, say,
the hydrodynamics of a Bose-Einstein condensate in the Gross-Pitaevskii
approximation. We derive the generalised conservation laws, and show
that the correction term admits an intuitive physical interpretation.
The detailed application of our results to group field theories, on
the other hand, will be presented elsewhere.

Our paper is structured as follows. An introduction of notation, conventions
and the geometrical space that we use throughout the paper is presented
in section \eqref{sec:Notation,-conventions-and}. In section \eqref{sec:Symmetries-and-continuity}
we provide the definition of the type of non-local action we are concerned
with, followed by the equations of motion and definition of symmetry
transformations. We derive our main result, that is the generalised
conservation law, in section \eqref{sec:Generalized-conservation-law}.
An application of our analysis to an explicit non-local field theory
is presented in section \eqref{sec:Physical-examples}. In the last
section we briefly review the derivation of Ward identities in the
functional integral formalism, and show the quantum counterpart of
our main result.

\section{Notation, conventions and basic definitions\label{sec:Notation,-conventions-and}}

In this section we introduce the notation, conventions and the main
definitions of the geometrical space defining the framework for our
analysis. In the concluding section of the paper, we will detail an
explicit example of a non-local field theory analyzed with the methods
developed in the bulk of the paper. We refer to the same example for
further clarification of our notation and definitions.

\subsection{Vector and jet bundles}

A vector bundle $\pi:\vb\to\bm$ with the projection $\pi$, the total
space $\vb$, the fiber $\fb$ and the base manifold $\bm$ is denoted
simply with $\vb$. We also assume that the base manifold $\bm$ is
orientable and is equipped with a (semi-)Riemannian metric. The covariant
derivative is denoted with $\cd{}$ and the subscript $\cd q$ denote
the point at which the derivative is taken.

Capital greek letters will denote sections on $\vb$. In a local trivialization
$U\subset M$, the section is given by $\Phi=\left\{ q,\phi\left(q\right)|q\in U\right\} $.
The correspondent functions $\phi:\bm\to\fb$ are called fields on
$\vb$ and are denoted with the correspondent lower case greek letter.
The space of smooth fields on $\vb$ is denoted by $\Gamma\left(E\right)$.
Additional restrictions on these sets, such as local fields around
$q\in M$, fields local in $U\subset M$ and compactly supported fields
in $U\subset M$, are denoted by $\Gamma_{p}\left(E\right),\,\Gamma_{U}\left(E\right)$
and $\Gamma_{U,C}\left(E\right)$, respectively.\footnote{Even in local field theories it is well known, that solutions to variational
problems might not be differentiable in general. In non local field
theories the situations is even more complicated. In this paper we
focus on the general procedure instead of treating the problem of
existents of solutions. For this reason we are not distinguishing
between smooth and $n$ times differentiable fields.}

The first jet bundle on $\vb$ is denoted $\jb$\footnote{For the convenience of the reader we briefly recall the notion of
the jet bundle. A first order jet bundle $\left(\jb,\pi_{j},\bm\right)$
over a vector bundle $\vb$ is a vector bundle which fibers are called
jet spaces $\js$. Let $\sim$ denote an equivalence relation on local
sections of $\vb$ where two local sections $\Phi,\Psi\in\Gamma_{p}\left(\vb\right)$
are called equivalent $\Phi\sim\Psi$ if 
\[
\phi\left(p\right)=\psi\left(p\right)\quad\text{and}\quad\cd p\phi=\cd p\psi\quad.
\]
The first order jet space at $p$ is a vector space of equivalence
classes $\js_{p}=\left\{ \left[\phi\right]_{p}\vert\phi\in\Gamma_{p}\left(\vb\right)\right\} $.
And an element of $\js_{p}$ can therefore be locally understood as
a triple 
\[
\Phi_{p}\oplus\cd p\Phi=\left(p,\phi\left(p\right),\cd p\phi\right)
\]
For this reason the jet space can be written as a set $\js_{q}=\left\{ \left(\phi\left(p\right),\cd p\phi\right)\vert\phi\in\Gamma_{p}\left(\vb\right)\right\} $.
The bundle projection $\pi_{j}:\jb\to\bm$ is then given by $\pi_{j}\left(\left(\phi\left(p\right),\cd p\phi\right)\right)=p.$
The first order jet space $J_{p}$ can be seen as a generalization
of the tangent space, where the equivalence class of curves is replaced
by the equivalence class of sections in $\vb.$ In our paper we often
refer to the component $\cd p\phi$ already as the jet space component
to distinguish it from the fiber element $\phi\left(p\right)$. }. Its fiber at point $q\in M$ is called the jet space $\jb_{q}=\js_{q}=\left\{ \left(\phi\left(p\right),\cd p\phi\right)\vert\phi\in\Gamma_{p}\left(\vb\right)\right\} $.
Sections on $\jb$ are denoted by $j\left(\Phi\right)$ with correspondent
fields $j\left(\phi\right)\simeq\cd{}\phi$ and are called prolonged
section and fields, respectively. In a local trivialization, $j\left(\Phi\right)=\left\{ \left(q,\phi\left(q\right),\cd q\phi\right)\vert q\in U\right\} $.

\subsection{Geometrical construction}

The base manifold of $\vb$ is the domain of the Lagrangian, itself
a functional of our physical fields. In local field theories it coincides
with the domain of the fields (and of their derivatives). For example,
in non-relativistic, Newtonian physics, and for a local field theory,
the base manifold is a product of time and space as $\bm=\mathbb{R}\times\mathbb{\mathbb{R}}^{3}$.
In a more general situation (e.g. non-spacetime based theories like
in group field theories), it can be a product of $N$ different manifolds
$\bm=M^{1}\times\cdots\times M^{N}$. In theories corresponding to
this more general situation, it is possible that the action and the
Lagrangian are local on some sub-manifolds $M^{i}$ and non-local
on the rest. Here by locality we mean that the arguments of the different
fields (and their derivatives) appearing in the Lagrangian, are identified
with one another when they refer to such sub-manifolds $M^{i}$. An
example of this situation is given again by non-relativistic field
theories as used in condensed matter theory, where usually the fields
appearing in the action differ by their spatial domain but are evaluated
at the same time. We call this distribution of non-localities the
\textit{combinatorial structure}, and say that a theory is \textit{trivially
non-local} if its Lagrangian depends on $N$ fully distinct points.
Local theories, in our setting, are those such that the fields appearing
in the Lagrangian have all their arguments identified, thus effectively
the Lagrangian depends on a single component $\bm$.

The underling jet bundle of a theory with trivial combinatorics is
a product of $N$ jet bundles over $\vb$. It is again a vector bundle
$\vb^{P}\simeq\vb^{\times N}$ with the base manifold $\bm^{P}\simeq\bm^{\times N}$
the fiber $\fb^{P}\simeq\fb^{\oplus N}$. The jet space at each point
is a direct sum of jet spaces in the same way as it is the fibers
$V^{P}$. That is for each $q\in\bm^{P}\simeq\left(q^{1},\cdots,q^{N}\right)$
$\js_{q}^{P}\simeq\js_{q^{1}}\oplus\cdots\oplus\js_{q^{N}}.$

For a theory with non trivial combinatorial structure the underling
base manifold $\BM$ is a submanifold of $\bm^{\times N}$. To be
able to study this case, we assume that $\BM$ can be isometrically
embedded in $\bm^{\times N}$ by an embedding $f$ that consists of
a combination of the following three maps: \\
 The identity 
\begin{eqnarray}
\mathds{1}: & M & \to M\\
 & q & \mapsto q,\nonumber 
\end{eqnarray}
the diagonal map 
\begin{eqnarray}
\text{Di}: & M & \to M\times M\\
 & q & \mapsto\left(q,q\right),\nonumber 
\end{eqnarray}
and the permutation map 
\begin{eqnarray}
\text{Per}: & M^{1}\times M^{2} & \to M^{2}\times M^{1}\\
 & \left(q^{1},q^{2}\right) & \mapsto\left(q^{2},q^{1}\right).\nonumber 
\end{eqnarray}
For example for the mentioned class of theories, which are local in
time and non-local in space, we get the embedding map $f:\mathbb{R}\times\mathbb{R}^{3}\times\mathbb{R}^{3}\hookrightarrow\mathbb{R}\times\mathbb{R}^{3}\times\mathbb{R}\times\mathbb{R}^{3}$
by setting 
\begin{equation}
f=\left(\mathds{1}\times\text{Per}\times\mathds{1}\right)\left(\text{Di}\times\mathds{1}\times\mathds{1}\right).\label{eq:an example for f}
\end{equation}
The space $\BM$ is therefore connected to $\bm^{P}$ by the embedding
$f$. We define our vector bundle $\VB$ for the general case as the
pull-back vector bundle of $\vb^{\times N}$. That is, 
\begin{equation}
\VB:=f^{*}\vb^{\times N}.
\end{equation}
The jet space at each point $q\in\BM$ is $\JS_{q}=\left\{ \CD q\ef=\left(\cd{q^{1}}\phi^{1},\cdots,\cd{q^{N}}\phi^{N}\right)\vert\phi\in\Gamma_{q}\left(E^{P}\right)\right\} $
where $\CD{}$ denotes the covariant derivative on $\BM$ given by
the pull back of $\cd{}^{\times N}$ and $q^{i}=\text{pr}^{i}\left(f\left(q\right)\right)$.

In cases where we need to use charts we denote the dimension of the
base manifold and the fiber $\dim\left(\bm\right)=\d{\bm}$ and $\dim\left(\fb\right)=\d{\fb}$,
respectively.

The canonical volume form on $M$ is denoted $\text{vol}$. In local
coordinates it has the usual expression $\text{\ensuremath{\vol}}=\sqrt{\vert g\vert}\text{d}x^{1}\wedge\cdots\text{d}x^{I}$.
The volume form on $\BM$ is denoted by $\VOL=f^{*}\vol^{\times N}$.
We will also sometimes sloppy denote $\vol_{q}$ ($\text{\ensuremath{\VOL}}_{\bar{q}}$)
to underline that the integration is (is not) performed over the manifold
whose points we label $q$. For example for trivial $f$\footnote{We use in this case a trivial $f$ only in order to simplify the notation.
A non trivial $f$ does not change the product structure of $\BM$
which is why the intuition about $\VOL_{\bar{q}}$ remains unchanged.} and an arbitrary, suitably integrable function $h$ we can choose
local coordinates $\left(x^{i}\right)_{n}$, with $i\in\left\{ 1,\cdots\d{\bm}\right\} $,
$n\in\left\{ 1,\cdots N\right\} $, of the manifold $\BM$, and use
our notation to define the following integrals
\begin{equation}
\int_{\BM/q^{i}}\,h\,\VOL_{\bar{q}^{i}}:=\int\,h\left(x_{1}^{1},\cdots,x_{\d{\bm}}^{1},\cdots,x_{1}^{i}\left(q\right),\cdots x_{\d{\bm}}^{i}\left(q\right),\cdots,x_{1}^{N},\cdots,x_{\d{\bm}}^{N}\right)\,\prod_{j\neq i}^{N}\sqrt{\left|g\right|}\,\prod_{m=1}^{\d{\bm}}\text{d}x_{m}^{j}\quad.
\end{equation}

Denoting the domain of the field $\phi^{i}$ by $\dom i$ (which is
a sub manifold of $\BM$) we will use the above notation and write
\begin{equation}
\int_{\BM/\dom i}\,h\,\VOL_{\bar{\dom i}}
\end{equation}
for integrals over the sub-manifold $\BM/\dom i$. To simplify the
notation we will symbolically use the notation of the delta $\delta^{i}\left(q\right)$
under the integral referring to the following identity 
\begin{equation}
\int_{\BM}\delta^{i}\left(q\right)\,h\,\text{Vol}:=\int_{\BM/\dom i}h\,\text{Vol}_{\bar{\dom i}}\quad,
\end{equation}
 the dependence of $h$ on the domain $\dom i$ is hereby set to a
fixed point $q\in\dom i$.

We will derive functions with respect to a parameter $\epsilon$.
In this case we always assume that the derivative is taken at the
point $\epsilon=0$ and don't write it explicitly to simplify readability.
In the single case where the derivative is not assumed to be at the
point zero we will explicitly denote it.

\subsection{Fiber derivatives \label{sub:Fiber-derivatives}}

It will be very useful to use notation that allows to avoid indices
and still suggests contractions of vector fields in the natural way.
In short, we treat derivatives as 1-forms on the appropriate space.
This allows us to write contractions as dual pairings without referring
to the index notation, improving readability.

Treating the Lagrangian $L$ as a function on three spaces - the base
manifold, the fiber and the jet - we introduce the common notions
of derivatives on each of these spaces.

Assuming that the points of the fiber $\FB$ and the jet $\JS$ are
given by a smooth field $\ef\in\Gamma\left(\VB\right)$ we can treat
the Lagrangian as a function $L:\BM\to\mathbb{R}$. Its derivative
at a point $q\in\BM$ is then a function 
\begin{equation}
\CD qL:T_{q}\BM\to\mathbb{R}\quad.
\end{equation}
In other words it is a one form on the tangent bundle $T\BM$.

On the other hand fixing the base point $q\in\BM$ we can treat the
Lagrangian as a function $L:\VB_{q}\times\JS_{q}\to\mathbb{R}$.

The derivative of $L$ in the fiber $\FB$ is denoted by 
\begin{align}
\CD{\FB}L\vert_{\Phi\left(q\right)}:T_{\Phi\left(q\right)}\VB_{q} & \to\mathbb{R}\quad.
\end{align}
It is seen as a 1-form on $T_{\Phi\left(q\right)}\VB_{q}\simeq\VB_{q}$.
We will often abandon labelling the base point and write $\CD{\FB}L\left(\ef\right)$
for $\ef\in\Gamma\left(\VB\right)$, meaning the 1-form $\CD{\FB}L$
contracted with the vector $\ef\left(q\right)\in\VB_{q}$. In local
coordinates $\left(x^{i},u^{j},u_{x^{i}}^{j}\right)_{n}$, for $i\in\left\{ 1,\cdots\d{\bm}\right\} $
$j\in\left\{ 1,\cdots,\d{\fb}\right\} $, $n\in\left\{ 1,\cdots N\right\} $\footnote{To keep the notation in the coordinate notation understandable we
assume a trivial $f$ as in the example above.}, and with $x_{n}^{j}=x_{n}^{j}\left(q\right)$, $u^{j,n}=\phi^{jn}\left(x_{1}^{j},\cdots,x_{N}^{j}\right)$
and $u_{x_{n}^{i}}^{j,n}=D_{x_{n}^{i}}\phi^{jn}\left(x_{1}^{j},\cdots x_{N}^{j}\right)$,
we get 
\begin{equation}
\CD{\FB}L\left(\phi\right)=\sum_{n=1}^{N}\sum_{s=1}^{\d{\fb}}\partial_{u^{s,n}}L\left(\left(x,u,u_{x}\right)\right)u^{s,n}\quad.
\end{equation}
It follows directly from the above coordinate representation that
the fiber derivative splits in a natural way into contractions on
each $E$ as 
\begin{equation}
\CD{\FB}L\left(\ef\right)=D_{V^{1}}L\left(\phi^{1}\right)+\cdots+D_{V^{n}}L\left(\phi^{n}\right)\quad.
\end{equation}

The derivative of $L$ in the jet $\JS$ is denoted by 
\begin{align}
\CD{\JS}L:T_{j\left(\Phi\right)}\JS_{q} & \to\mathbb{R}\quad,
\end{align}
again with $T_{j\left(\Phi\right)}\JS_{q}\simeq\JS_{q}$. It is a
function on the differentials of fields $\CD q\ef$. The later, however,
are linear functions on $T_{q}\BM$ which are isomorphic to sections
of the product bundle $E\otimes\left(T\BM\right)^{*}$. For this reason
$\CD{\JS}L$ is isomorphic to sections of $E^{*}\otimes T\BM$. We
mean exactly this natural contraction when we write $\CD{\JS}L\left(\CD q\ef\right)$.
In coordinates, we can explicitly write 
\begin{equation}
\CD{\JS}L\left(D\phi\right)=\sum_{n=1}^{N}\sum_{s=1}^{\d{\fb}}\sum_{i=1}^{\d{\bm}}\partial_{u_{x_{n}^{i}}^{s,n}}L\left(\left(x,u,u_{x}\right)\right)\cdot u_{x_{n}^{i}}^{s,n}\quad.
\end{equation}
This contraction splits in the same way as above in 
\begin{equation}
\CD{\JS}L\left(D\phi\right)=D_{J^{1}}L\left(D\phi^{1}\right)+\cdots+D_{J^{n}}L\left(D\phi^{n}\right)\quad.
\end{equation}
It will be useful to integrate by parts the above relation to get
\begin{equation}
\CD{\JS}L\left(D\ef\right)=\text{tr}\left[\CD{}\left(\CD{\JS}L\right)\right]\left(\ef\right)-\text{tr}\left[\CD{}\left(\CD{\JS}L\left(\ef\right)\right)\right]\quad,
\end{equation}
where the trace is understood as the contraction of the 1-form coming
from $\CD{}$ with the vector part $T_{p}\BM$ of the $\CD{\JS}L$.
In local coordinates the trace becomes 
\begin{equation}
\text{tr}\left[\CD{}\left(\CD{\JS}L\left(\ef\right)\right)\right]=\sum_{n=1}^{N}\sum_{i=1}^{\d{\bm}}\sum_{j=1}^{\d{\fb}}\partial_{x_{n}^{i}}\partial_{u_{x_{n}^{i}}^{j,n}}L\cdot u^{j,n}\quad.
\end{equation}
which is the covariant definition of divergence. We write 
\begin{equation}
\CD{\JS}L\left(\CD q\ef\right)=\text{div}\left[\CD{\JS}L\right]\left(\ef\right)-\text{div}\left[\left(\CD{\JS}L\left(\phi\right)\right)\right]\quad.
\end{equation}
Note that also the right-hand-side has a natural splitting into fibers
as 
\begin{equation}
\text{div}\left(\CD{\JS}L\right)\left(\ef\right)=\text{div}_{\dom 1}\left[D_{J^{1}}L\right]\left(\phi^{1}\right)+\cdots+\text{div}_{\dom N}\left[D_{J^{n}}L\right]\left(\phi^{n}\right)\quad,\label{eq:Fibration of the divergence}
\end{equation}
where we denote $\text{div}_{\dom i}$ the divergence on the single
domain of $\phi^{i}$. As we can see the fiber derivatives split in
a natural way into a sum indexed by the field index $j\in\left\{ 1,\cdots,N\right\} $.
However, notice that the divergence of a general vector field $X\in T\BM$
does not split in the above way.

To conclude the section we present a short table (\eqref{tab:An-overview-of})
of symbols that we use throughout the paper.

\begin{table}[H]
\begin{centering}
\begin{tabular}{>{\centering}p{1.5cm}|>{\raggedright}p{5cm}|>{\centering}p{1.5cm}>{\raggedright}p{5cm}}
\multicolumn{2}{c|}{\textbf{Single vector bundle}} & \multicolumn{2}{c}{\textbf{Pull back bundle}}\tabularnewline
\hline 
$\vb$  & Total space  & $\VB$  & Total space \tabularnewline
\hline 
$\bm$  & Base manifold  & $\BM$  & Base manifold \tabularnewline
\hline 
$\fb$  & Single fiber  & $\FB$  & Fiber \tabularnewline
\hline 
$\pi$  & Projection on $\bm$  & $\pi^{\VB}$  & Projection on $\BM$\tabularnewline
\hline 
$\js$  & Jet space  & $\JS$  & Jet space \tabularnewline
\hline 
$\jb$  & Jet bundle  & $\JB$  & Jet bundle \tabularnewline
\hline 
$\cd{}$  & Covariant derivative  & $\CD{}$  & Covariant derivative\tabularnewline
\hline 
$\cd{\fb}$  & Fiber derivative in $\fb$  & $\CD{\FB}$  & Fiber derivative in $\FB$\tabularnewline
\hline 
$\cd{\js}$  & Fiber derivative in $\js$  & $\CD{\JS}$  & Fiber derivative in $\JS$\tabularnewline
\hline 
$\vol$  & Volume element on $\bm$  & $\VOL$  & Volume element on $\BM$\tabularnewline
\hline 
 &  & $\dom i$  & Domain of the fiber $\fb^{i}\subset\FB$\tabularnewline
\hline 
$\d{\bm}$  & Dimension of $\bm$  & $\d{\BM}$  & Dimension of $\BM$\tabularnewline
\hline 
$\d{\fb}$  & Dimension of $\fb$  & $N$  & Number of non local points\tabularnewline
\hline 
$\phi$  & Smooth fields in $\Gamma\left(\vb\right)$  & $\ef$  & Smooth fields in $\Gamma\left(\VB\right)$\tabularnewline
\hline 
\end{tabular}
\par\end{centering}

\protect\caption{An overview of the most used symbols in this paper \label{tab:An-overview-of}}
\end{table}

\section{Symmetries and continuity equations for non-local actions\label{sec:Symmetries-and-continuity}}

In this section we give the definition of the non-local field theories,
described by non-local action functionals, that we consider in this
paper. Loosely speaking, we call \textit{non-local} an action functional
that is defined by a Lagrangian, which depends on $N$ copies of a
jet bundle. Theories described by this kind of actions often appear
when the physical system consists of many particles, their interaction
can not be neglected and is not of contact type, in the classical
limit. They appear in various branches of physics such as hydrodynamics,
condensed matter and solid state physics. In fact, an example of non-local
theories of the above type is given by effective theories, which arise
when one disregards some microscopic degrees of freedom to obtain
a theory of fewer variables. The price to pay is often the emergence
of non-locality of the type we study. The most prominent examples
of this type are models with Coulomb interaction, which are effective
theories of an underlying local (and relativistic) electrodynamics.
They also appear in some models of quantum gravity, although of course
the physical interpretation is very different. We also point out,
that although non local theories often arise from coarse graining
of microscopic local theories, it is (to our knowledge) still unknown
if any non-local theory admits an underlying, physically adequate,
local description.

We begin with the definition of a non-local geometrical Lagrangian
and action, as the most natural extension of the usual local ones
\cite{Abraham:2008aa,Olver:1998aa}. From them, we obtain what we
call the physical non-local Lagrangians. In the following two subsections
we then apply the variational principle and the theory of Lie groups
to obtain the equations of motions and local symmetry statements.

\subsection{The non local action}

In this subsection we introduces the various definitions for the non-local
geometrical and physical actions. We will make a distinction between
geometrical and physical quantities with an superscript $G$ or $P$
respectively. The motivation for this separation will become apparent
at the end of this section. We begin with the definition of non local
geometrical Lagrangian. 
\begin{defn}
Let $\JB$ be a jet bundle over the vector bundle $\VB$. A map $L^{G}:\JB\to\mathbb{R}$
is called a geometrical Lagrangian if it is differentiable in $\BM,\FB,\JS$
with respect to our definitions of fiber derivatives in the previous
section. 
\end{defn}
Recall that the fiber $\FB$ is isomorphic to the direct sum of $N$
fibers $V$. We call $N$ the \textsl{degree of non locality}, with
$N=1$ defining a local theory. Just as in the local case we can define
the non-local geometrical action by integrating the Lagrangian over
a region of the base manifold. 
\begin{defn}
Let $\Omega\subset\BM$ be an open region. A non-local geometrical
action is a functional $S_{\Omega}^{G}:\Gamma_{\Omega}\left(\VB\right)\to\mathbb{R}$
defined by the Lagrangian as follows 
\begin{equation}
S_{\Omega}^{G}\left[\ef\right]=\int_{\Omega}L^{G}\left(j\left(\Phi\right)\right)\,\text{Vol}\quad.
\end{equation}

\end{defn}
Since a field at each point $\VB$ has a direct sum structure $\ef\left(q\right)=\left(\phi^{1},\cdots,\phi^{N}\right)\left(f\left(q\right)\right)$
the above action corresponds to a situation where the fields cary
labels, which makes them distinguishable. In this case different fields
are evaluated at different points. Such theories are sometimes called
colored. However, in our analysis we are mainly interested in non-local
physical systems with a single field.

To define the Lagrangian on a single field we use the diagonal map
\begin{eqnarray}
\imath_{N}: & \Gamma\left(\vb\right) & \to\bigoplus_{i=1}^{N}\Gamma\left(\vb\right)\simeq\Gamma\left(\VB\right)\\
 & \phi & \mapsto\left(\phi,\cdots,\phi\right)\quad,\nonumber 
\end{eqnarray}
which can be seen as an embedding of $\Gamma\left(\vb\right)$ into
the larger space $\Gamma\left(\VB\right)$. We will sometimes write
$\imath_{N}\left(\phi\right)=\sum_{i}^{N}\phi\otimes e_{i}$ with
the standard basis $\left\{ e_{i}\right\} $ in $\mathbb{R}^{N}$.
Finally we define the non local physical Lagrangian as follows. 
\begin{defn}
Let $L^{G}$ be a geometrical Lagrangian on $\JB$ and $\imath_{N}$
be the diagonal embedding map. The corresponding non-local physical
Lagrangian $L^{P}:\BM\times\Gamma\left(\vb\right)\to\mathbb{R}$ is
defined point-wise by 
\begin{equation}
L^{P}\left(q,\phi\right)=L^{G}\left(q,\left[\imath_{N}\phi\right]\left(q\right),\CD q\left[\imath_{N}\phi\right]\right)\quad.
\end{equation}

\end{defn}
It is the physical Lagrangian which has a more straightforward physical
interpretation (hence the chosen name), even though the geometrical
interpretation is much clearer by the geometrical Lagrangian. Because
of this, we will first reformulate the physical conditions such as
equations of motion in terms of the geometrical Lagrangian and use
then the well established theory of Lie groups to obtain local symmetry
conditions.

First, we introduce the non-local physical action, whose variation
will lead to the equations of motion. 
\begin{defn}
Let $U\subset M$ be open, and $\Omega=f^{-1}\left(U^{\times N}\right)\subset\BM$.
A non-local physical action is a functional $S_{\Omega}^{P}:\Gamma_{U}\left(E\right)\to\mathbb{R}$
defined by the physical Lagrangian as 
\begin{equation}
S_{\Omega}^{P}\left[\phi\right]=\int_{\Omega}L^{P}\left(\cdot,\phi\right)\,\text{Vol}\quad.
\end{equation}

\end{defn}
Using the definition of the physical Lagrangian and the geometrical
action we see the connection between the physical and geometrical
quantities 
\begin{equation}
S_{\Omega}^{P}=S_{\Omega}^{G}\circ\imath_{N}\quad.
\end{equation}
In this formulation we can state that the main formal reason for the
failure of the standard Noether theorem in the non-local case is the
difference between the physical and geometrical action, as we will
see shortly. Notice that, indeed, in the local case ($N=1$) the inclusion
map $\imath_{N}$ is the identity, which also clarifies why the connection
of geometrical symmetries and physical equations of motion is natural.
In the non-local case, on the other hand, we will be facing a problem
in connecting the equations of motion to the symmetry properties of
the geometrical functional.

In the next section we present the equations of motion in terms of
the geometrical action and Lagrangian.

\subsection{Non local equations of motion}

In this section we apply the variational principle to the non-local
physical action $S^{P}$ and obtain the (semi-local) equations of
motion. 
\begin{lem}
\label{Naturality of i}Let $\varphi_{\epsilon}\in\Gamma\left(\vb\right)$
be a family of smooth fields, differentiable in the parameter $\epsilon$.
The diagonal inclusion map $\imath_{N}$ commutes with the partial
derivative in $\epsilon$ in the following sense, 
\begin{equation}
\partial_{\epsilon}\,\imath_{N}\left(\varphi_{\epsilon}\right)=\imath_{N}\,\partial_{\epsilon}\left(\varphi_{\epsilon}\right)\quad.
\end{equation}
\end{lem}
\begin{proof}
By direct computation we obtain 
\begin{equation}
\partial_{\epsilon}\,\imath_{N}\left(\varphi_{\epsilon}\right)=\partial_{\epsilon}\left(\varphi_{\epsilon},\cdots,\varphi_{\epsilon}\right)=\left(\partial_{\epsilon}\varphi_{\epsilon},\cdots,\partial_{\epsilon}\varphi_{\epsilon}\right)=\imath_{N}\left(\partial_{\epsilon}\varphi_{\epsilon}\right)\quad.
\end{equation}

\end{proof}
Let $\varphi\in\Gamma_{U,C}\left(\vb\right)$ be an arbitrary, compactly
supported, smooth field in $U$ and let $\Omega\subset\mbox{\ensuremath{\BM}}$
such that $f\left(\Omega\right)=U^{\times N}$. The variation of $S_{\Omega}^{P}$
is given by 
\begin{align}
\text{d}S^{P}\vert_{\phi}\left(\varphi\right) & =\partial_{\epsilon}S_{\Omega}^{P}\left(\phi+\epsilon\varphi\right)\\
 & =\int_{\Omega}\,\partial_{\epsilon}L^{P}\left(\cdot,\phi+\epsilon\varphi\right)\,\VOL\nonumber \\
 & =\int_{\Omega}\,\partial_{\epsilon}L^{G}\left(\cdot,\imath_{N}\left[\phi+\epsilon\varphi\right],\CD{}\left(\imath_{N}\left[\phi+\epsilon\varphi\right]\right)\right)\,\text{\ensuremath{\VOL}}\quad.\nonumber 
\end{align}
By lemma \eqref{Naturality of i} we get 
\begin{equation}
\text{d}S^{P}\vert_{\phi}\left(\varphi\right)=\int_{\Omega}\,\CD{\FB}L^{G}|_{j\left(\Phi_{D}^{P}\right)}\left[\imath_{N}\varphi\right]+\CD{\JS}L^{G}|_{j\left(\Phi^{P}\right)}\left[\CD{}\left(\imath_{N}\varphi\right)\right]\,\VOL\quad.
\end{equation}
And by partial integration we obtain 
\begin{equation}
\text{d}S^{P}\vert_{\phi}\left(\varphi\right)=\int_{\Omega}\,\left[\CD{\FB}L^{G}-\text{div}\left(\CD{\JS}L^{G}\right)\right]_{j\left(\Phi^{P}\right)}\left[\imath_{N}\varphi\right]\,\VOL\quad.
\end{equation}
The right-hand-side defines the Euler-Lagrange equations by the extremality
condition 
\begin{equation}
\text{d}S\vert_{\phi}\left(\varphi\right)=0\qquad\forall\varphi\in\Gamma_{U,C}\left(E\right)\quad,
\end{equation}
We define the Euler 1-form as $E:=\left[\CD{\FB}L^{G}-\text{div}\left(\CD{\JS}L^{G}\right)\right]$
(not to be confused with the vector bundle $E$ ), and write 
\begin{equation}
\text{d}S\vert_{\phi}\left(\varphi\right)=\int_{\Omega}E_{\left(\Phi^{P}\right)}\circ\imath_{N}\left(\varphi\right)\,\VOL\quad.
\end{equation}
In the following we will abandon the label for the base point as well
as the notation of the integral domain $\Omega$ and simply write
$\text{d}S\left(\varphi\right)=\int_{\Omega}E\circ\imath_{N}\left(\varphi\right)$.
This equation gives rise to semi local equations of motion by the
fundamental lemma of variations, that we state below for convenience
(for further references see for example \cite{Giaquinta:1996aa}). 
\begin{lem}[Fundamental lemma of variation]
Let $f$ be a continuous, real-valued function on some region $U\subset\mathbb{R}^{m}$,
and suppose that 
\begin{equation}
\int_{U}\,f\left(x\right)\varphi\left(x\right)\,\text{d}x=0\quad,
\end{equation}
holds for all $\varphi\in C_{C}^{\infty}\left(U\right)$ with $\varphi\geq0$.
Then 
\begin{equation}
f\left(x\right)=0\quad,
\end{equation}
for all $x\in U$. 
\end{lem}
A corollary of this lemma provides the semi-local equations of motion. 
\begin{cor}
Let $E$ be a one form on $\VB$. Then we can locally write $E=\sum_{i=1\,}^{N}f^{i,j}\otimes v^{j}\otimes e^{i}$
with continues functions $f^{i,j}$ on $\Omega\times\FB$, one forms
$v^{j}$ on $V$ and dual vectors $e^{i}$ given by $e^{i}\left(e_{j}\right)=\delta_{i,j}$
for the standard basis $\left\{ e_{j}\right\} _{j=1\cdots N}$ in
$\mathbb{R}^{N}$. If 
\begin{equation}
\int_{\Omega}\left(E\circ\imath_{N}\right)\left(\varphi\right)\VOL=0\quad,
\end{equation}
for all $\varphi\in\mathcal{C_{C}^{\infty}}\left(U\right)$, then
we have for all $q\in U$ 
\begin{equation}
\int_{\Omega/\dom 1}\,f^{1}\,\VOL_{\bar{\dom 1}}+\cdots+\int_{\Omega\backslash\dom N}f^{N}\,\VOL_{\bar{\dom N}}=0\quad,
\end{equation}
where the domain over which the integration is not performed is set
to $q$. \end{cor}
\begin{proof}
Applying the definitions we get 
\begin{align}
\int_{U^{\times N}}E\circ\imath_{N}\left(\varphi\right)\,\VOL & =\int_{\Omega}E\left[\imath_{N}\left(\varphi\right)\right]\,\VOL\\
 & =\sum_{j=1}^{v}\sum_{i=1}^{n}\int_{\Omega}\left(f^{i,j}\otimes v^{j}\otimes e^{i}\right)\left[\varphi^{j}\otimes v_{j}\otimes e^{i}\right]\nonumber \\
 & =\sum_{j=1}^{v}\int_{\Omega}\left(f^{1,j}\cdot\varphi^{j}+\cdots+f^{n,j}\cdot\varphi^{j}\right)\,\VOL\nonumber \\
 & =\sum_{j=1}^{v}\int_{U}\left\{ \sum_{i=1}^{n}\int_{\Omega\backslash\dom i}\,f^{i,j}\,\VOL_{\bar{\dom i}}\right\} \varphi^{j}\,\text{\ensuremath{\vol}}\quad.\nonumber 
\end{align}
By assumption, the functions $f^{j}:=\left\{ \sum_{i=1}^{n}\int_{U^{\times N-1}}\,f^{i,j}\,\VOL\right\} $
are continuous on $U$. Since $\varphi^{j}$ are independent for different
$j$ the statement follows from the fundamental lemma. 
\end{proof}
The semi-local equations of motion become 
\begin{equation}
EL\left[\cdot\right]\left(q\right)=\sum_{i}^{N}\int_{\Omega/\dom i}\,E^{i}\left(q\right)\circ\left[\imath_{N}\cdot\right]\,\VOL\quad,
\end{equation}
with the 1-forms $E^{i}=\left[D_{V^{i}}L^{G}-\text{div}_{i}\left(D_{J^{i}}L^{G}\right)\right]$.
Here the point $q\in\bm$ is subsequently identified with each domain
$\dom i$ on which the integration is not performed. Notice that these
equations are local in $q\in M$ but also depend on the value of the
field at other points on $U$ (which follows from the non-local nature
of the action). For this reason we call them \lq semi-local\rq.

For the sake of readability we will write

\begin{equation}
EL\left[\cdot\right]\left(q\right)=\sum_{i=1}^{N}\int_{\Omega}\,\delta^{i}\left(q\right)\,E^{i}\left[\imath_{N}\cdot\right]\,\VOL\quad,\label{eq:Equations of motion}
\end{equation}
using the delta notation introduced earlier. Note that due to the
above description $EL$ is a 1-form on the fiber of $\VB$, or equivalently
on vertical vector fields. In the local case $N=1$ which implies
$\imath_{N}=\text{id}$ and $\Omega=U$, the sum vanishes and we obtain
the local equations of motion 
\begin{equation}
EL\left[\cdot\right]\left(q\right)=E\left[\cdot\right]\left(q\right)=\left\{ D_{V}L-\text{div}\left[D_{J}L\right]\right\} _{\Phi\left(q\right)}\left[\cdot\right]\quad.\label{eq:Local equations of motion}
\end{equation}

Let us conclude this section with a couple of additional remarks. 
\begin{itemize}
\item The domain of integration for the physical action is chosen to be
$f^{-1}\left(U^{\times N}\right)$. This is important since the case
with a general $\Omega\subset\BM$ generates additional technical
problems, which we do not address here. The solution space of the
theory gets additional restrictions on the boundaries of each $\Omega$.
In a simple case $f=\mathds{1}$ and for $\Omega=U^{1}\times U^{2}\times\cdots\times U^{N}$
with $U^{1}\subset U^{2}\subset\cdots U^{N}$ the variation of the
action has to happen in the direction of smooth fields $\eta$ that
are compactly supported on all $U^{i}$ in $\Omega$. That implies
that $\eta$ has to vanish on at least $N$ open sets inside $U^{N}$.
This condition would affect the subsequent analysis. 
\item In a non-local theory as we have defined it, adding to the Lagrangian
the total divergence on $\BM$ of a suitable tuple function $P$ can
change the equations of motion and therefore lead to a non-equivalent
Lagrangian. This can be easily seen for the case $f=\mathds{1}$.
Due to Stokes' theorem, the action will get an additional term which
depends only on the boundary of $\Omega=U^{\times N}$. However, the
variation is not assumed to vanish on the boundary of $\Omega$ but
just on the boundary of $U$. The boundary of $\Omega$ is instead
proportional to $\partial U^{\times N}\simeq\partial U\times U^{N-1}$,
while the test field $\eta$ is vanishing only on $\partial U$, and
therefore the added function $P$ is affected by the variation. This
feature, which may lead to complications, distinguishes in general
the non-local case from the local one. 
\end{itemize}

\subsection{Symmetry of the non local action}

In this section we define the notion of a symmetry for the non-local
geometrical action. This notion is global, meaning that it is a statement
about integral quantities. However, we show that this notion is equivalent
to a fully local condition in exactly the same way as it is in the
local case.

A vector field $X\in T\VB$ induces an action of a one parameter group
$G$ on $\VB$ by translating the points on $\VB$ along the flow
$c_{\epsilon}$ of $X$. The vector field $X$ is then homomorphic
to the Lie algebra $\mathfrak{g}$ of the group $G$ and is called
the infinitesimal generator of the group action. Specifically, let
$\cdot$ denote the action of the group $G$ on $\VB$. Then for a
group element $g_{\epsilon}=\exp\left(\epsilon v\right)$ with $v\in\textfrak{g}$
and $\Phi\left(q\right)$ a point on $\VB$ the corresponding vector
field is 
\begin{equation}
\partial_{\epsilon}\,g_{\epsilon}\cdot\Phi\left(q\right)=X_{\Phi\left(q\right)}=\partial_{\epsilon}\,c_{\epsilon}\left(\Phi\left(q\right)\right)\quad,\label{eq:Infenitesimal generators}
\end{equation}
where $c_{\epsilon}$ denotes the flow of $X$ on $\VB$.

By the bundle projection $\pi^{\VB}$ we can split every vector field
$X$ as $X=X_{\BM}+X_{\FB}$ where $X_{\BM}=\text{d}\pi^{\VB}\left(X\right)$
and $X_{\FB}=X-X_{\BM}$. At each point the vector field $X_{\BM}$
is tangent to the base manifold and $X_{\FB}$ is parallel to the
fiber. We call the corresponding flows $c_{\BM\epsilon}$ and $c_{\FB\epsilon}$,
respectively. At every point we get 
\begin{equation}
\partial_{\epsilon}c_{\epsilon}=X=X_{\BM}+X_{\FB}=\partial_{\epsilon}c_{\BM\epsilon}+\partial_{\epsilon}c_{\FB\epsilon}=\partial_{\epsilon}\left(c_{\BM\epsilon},c_{\FB\epsilon}\right)\quad.
\end{equation}
Integrating from 0 to $\epsilon$ and using $c_{0}=c_{\BM0}=c_{\FB0}=\mathds{1}$
we get locally 
\begin{equation}
c_{\epsilon}=\left(c_{\BM\epsilon},c_{\FB\epsilon}\right)\quad.
\end{equation}
Notice that $c_{\BM\epsilon}:\VB\to\BM$ and $c_{\VB\epsilon}:\VB\to\FB$.
However, fixing a section $\Phi$ allows us to view the transformation
maps as 
\begin{align}
c_{\BM\epsilon}:\BM & \to\BM & c_{\VB\epsilon}: & \BM\to\FB\quad.
\end{align}
Therefore, locally for any $\Phi\left(q\right)=\left(q,\ef\left(q\right)\right)$
and any fixed $\ef\in\Gamma\left(\VB\right)$ the group action can
be split as 
\begin{equation}
c_{\epsilon}\left(q,\ef\left(q\right)\right)=\left(c_{\BM\epsilon}\left(q\right),c_{\FB\epsilon}\circ\ef\circ c_{\BM\epsilon}^{-1}\circ c_{\BM\epsilon}\left(q\right)\right)\quad.
\end{equation}
Calling the new points $c_{\BM\epsilon}\left(q\right)=q_{\epsilon}$
and the transformed fields as $\ef_{\epsilon}:=c_{\FB\epsilon}\circ\ef\circ c_{\BM\epsilon}^{-1}$
we obtain the transformation of the section as 
\begin{align}
\Phi_{\epsilon}\left(q_{\epsilon}\right) & =\left(q_{\epsilon},\ef_{\epsilon}\left(q_{\epsilon}\right)\right)\quad.\label{eq:Transformations of fields and sections}
\end{align}
The transformation of $\ef$ induce a transformation of the prolonged
fields $j\left(\ef\right)$, which we call $c_{J\epsilon}$, given
by 
\begin{equation}
c_{J\epsilon}\circ j\left(\ef\right)\circ c_{\BM\epsilon}^{-1}=j\left(c_{\FB\epsilon}\circ\ef\circ c_{\BM\epsilon}^{-1}\right)=j\left(\ef_{\epsilon}\right)\quad.
\end{equation}
We note that the section $\Phi_{\epsilon}$ can be treated as a graph
on the transformed manifold as in Eq. \eqref{eq:Transformations of fields and sections}
or equivalently as a function on the non-transformed base manifold
as 
\begin{align}
\Phi_{\epsilon}\left(q\right) & =\left(c_{\BM\epsilon}\left(q\right),\ef_{\epsilon}\circ c_{\BM\epsilon}\left(q\right)\right)\quad.\label{eq:Second view on the transforamtions of sections}
\end{align}
The section of the jet bundle becomes 
\begin{equation}
j\left(\Phi_{\epsilon}\right)\left(q\right)=\left(c_{\BM\epsilon},\ef_{\epsilon}\circ c_{\BM\epsilon},j\left(\ef_{\epsilon}\right)\circ c_{\BM\epsilon}\right)\left(q\right)\quad.
\end{equation}
The variation of the action along $X$ is then given by

\begin{align}
\text{d}S\vert_{\phi^{P}}\left(X\right)=\partial_{\epsilon}\tilde{S}\left[\ef_{\epsilon}\right] & =\partial_{\epsilon}\int_{c_{\BM\epsilon}\left(\Omega\right)}L\left(j\left(\Phi_{\epsilon}\right)\right)\,\VOL\\
 & =\int_{\Omega}\partial_{\epsilon}\left\{ L\left(j\left(\Phi_{\epsilon}\right)\right)\,c_{M\epsilon}^{*}\VOL\right\} \quad.\nonumber 
\end{align}
It turns out to be very convenient to define the characteristic vector
field of $X$ as 
\begin{align}
X_{Q} & :=\partial_{\epsilon}\ef_{\epsilon}=X_{\FB}-\CD{}\ef\left(X_{\BM}\right)\quad,
\end{align}
and call the coefficient of $X_{Q}$ the characteristic $Q$ of $X$.
Notice that $X_{Q}$ is parallel to the fiber, since 
\begin{equation}
\text{d}\pi^{\VB}\left(X_{Q}\right)=0\quad.
\end{equation}

\begin{defn}
A local group $G$ acting on $\VB$ is called a symmetry group of
the non-local action $S_{\Omega}$, if a transformation $c\in G$,
written in local trivialization as $\left(c_{\BM},c_{\FB}\right)$,
satisfies 
\begin{align}
\tilde{S}\left(c\circ\ef\right) & :=\int_{c_{\BM}\left(\Omega\right)}\,L\left(c\circ j\left(\Phi\right)\right)\,\text{Vol}=\int_{\Omega}\,L\left(j\left(\Phi\right)\right)\,\text{Vol}=S\left(\ef\right)\quad.
\end{align}

\end{defn}
We now adopt the well known theorem that connects the symmetry transformation
to a local change of the Lagrangian. The mere difference of the local
case to our situation is that the vector bundle in the local case
is a single $\vb$ whereas in our case it is $\VB$. It is therefore
obvious that geometrically these two cases are not different which
is why the theorem carries over one to one. Nevertheless, we present
the theorem as well as its proof for the convince of the reader to
adopt it to our notation. The proof below follows closely the proof
in \cite{Abraham:2008aa} and partially in \cite{Olver:1998aa}. 
\begin{thm}
\label{Local and global symmetry} A connected one-parameter group
$G$ of transformations acting on $\VB$ is a symmetry group of the
action $S$ if and only if 
\begin{equation}
\left[\CD{}L\left(X_{\BM}\right)+\CD{\FB}L\left(X_{Q}\right)+\CD{\JS}L\left(DX_{Q}\right)+L\,\text{div}\left(X_{\BM}\right)\right]=0\quad,\label{eq:local symmetry equation}
\end{equation}
where $X$ is the infinitesimal generator of the group.\end{thm}
\begin{proof}
If $c_{\epsilon}=\left(c_{\BM\epsilon},c_{\FB\epsilon}\right)$ is
a symmetry transformation then 
\begin{equation}
\partial_{\epsilon}\tilde{S}\left(\phi_{\epsilon}\right)=0\quad,
\end{equation}
and by definition of the non-local action we get 
\begin{equation}
0=\partial_{\epsilon}\int_{c_{\BM\epsilon}\left(\tilde{\Omega}\right)}\,L\left(j\left(\Phi_{\epsilon}\right)\right)\,\text{\ensuremath{\VOL}}=\int_{\tilde{\Omega}}\,\partial_{\epsilon}\left[L\left(j\left(\Phi_{\epsilon}\right)\right)\,c_{M\epsilon}^{*}\text{\ensuremath{\VOL}}\right]\quad,
\end{equation}
for all $\tilde{\Omega}\subset\Omega$. This implies the local statement
\begin{equation}
\partial_{\epsilon}\left[L\left(j\left(\Phi_{\epsilon}\right)\right)\,c_{M\epsilon}^{*}\VOL\right]=0\quad.\label{eq:local symmetry equation - proof}
\end{equation}
On a manifold with a volume element the divergence can be written
as 
\begin{equation}
\partial_{\epsilon}\,\left(c_{M\epsilon}^{*}\VOL\right)=\text{div}\left(X_{M}\right)\,\VOL\quad.
\end{equation}
With this definition and the chain rule we obtain equation \eqref{eq:local symmetry equation}.

Conversely, if equation \eqref{eq:local symmetry equation} is satisfied
everywhere then the following is also true 
\begin{equation}
\left[\CD{}L\left(X_{\BM}\right)+\CD{\FB}L\left(X_{Q}\right)+\CD{\JS}L\left(DX_{Q}\right)+L\,\text{div}\left(X_{\BM}\right)\right]_{\left(\Phi_{\epsilon}\right)}\,c_{M\epsilon}^{*}\text{Vol}=0\quad,\label{eq:finite local transformation}
\end{equation}
where the equation in the brackets is to be taken at the point $\Phi_{\epsilon}$
for a small but finite $\epsilon$. Then the equation \eqref{eq:finite local transformation}
is a differential of the equation \eqref{eq:local symmetry equation - proof}
which we can write as 
\begin{equation}
\partial_{\epsilon}\left[L\left(j\left(\Phi_{\epsilon}\right)\right)\,c_{M\epsilon}^{*}\text{Vol}\right]\vert_{\epsilon\,\text{finite}}=0\quad,
\end{equation}
which leads to 
\begin{equation}
\partial_{\epsilon}\tilde{S}\left(\ef_{\epsilon}\right)=0\quad.
\end{equation}
Integrating this equation from 0 to $\epsilon$ and using the fact
that the $c_{0}=\mathds{1}$, we get 
\begin{equation}
\tilde{S}\left(\ef_{\epsilon}\right)=S\left(\ef\right)\quad,
\end{equation}
for a transformation $c_{\epsilon}$ sufficiently near the identity.
However, since every connected one-dimensional subgroup is generated
by a transformation of the form $c_{\epsilon}=\exp\left(\epsilon X\right)$
the above statement holds everywhere, which concludes the proof. 
\end{proof}
For local theories the corollary of the above theorem is known as
the Noether theorem \cite{Noether:1971aa,Olver:1998aa}. In our case,
however, it is valued on the product bundle $\VB$ and it will now
become apparent that the Noether theorem fails. 
\begin{cor}
\label{cor:Local symmetry with Euler one form} Let $X$ be a vector
field on $\VB$ with characteristic $X_{Q}$ and projection on the
base manifold $X_{\BM}$. Then $X$ induces a symmetry transformation
iff 
\begin{equation}
\text{div}\left(\CD{\JS}L\left(X_{Q}\right)\right)+\text{div}\left(L\cdot X_{\BM}\right)+E\left(X_{Q}\right)=0\quad.\label{eq:Symmetry equation 2}
\end{equation}
With the Euler one form defined as above, $E:=\CD{\FB}L-\text{div}\left(\CD{\JS}L\right)$.\end{cor}
\begin{proof}
From theorem \eqref{Local and global symmetry}, $X$ is a symmetry
iff 
\begin{equation}
\left[\CD{}L\left(X_{\BM}\right)+\CD{\FB}L\left(X_{Q}\right)+\CD{\JS}L\left(\CD{}X_{Q}\right)+L\,\text{div}\left(X_{\BM}\right)\right]\VOL=0\quad.\label{eq:Symmetry variation of the Lagrangian}
\end{equation}
By partial integration as described in section \ref{sub:Fiber-derivatives}
we obtain 
\begin{equation}
\text{div}\left(\CD{\JS}L\left(X_{Q}\right)\right)+\text{div}\left(L\cdot X_{\BM}\right)+\left[\CD{\FB}L-\text{div}\left(\CD{\JS}L\right)\right]\left(X_{Q}\right)=0\quad.\label{eq:Symmetry equation}
\end{equation}

\end{proof}
In the local case the Euler one form coincides with the equations
of motion Eq. \eqref{eq:Local equations of motion}, i.e. $E=EL$
and equation \eqref{eq:Symmetry equation 2} becomes the well known
Noether identity, 
\begin{equation}
\text{div}\left(\CD{\JS}L\left(X_{Q}\right)+L\cdot X_{\BM}\right)+EL\left(X_{Q}\right)=0\quad.
\end{equation}
In the non local case, however, we have $E\neq EL,$ which makes it
difficult to connect it to physically relevant quantities and leads
to the violation of the Noether theorem.

In general, equation \eqref{eq:Equations of motion} and \eqref{eq:Symmetry equation 2}
are very different, since the first one applies only on diagonal vector
fields, whereas the second is true for a generic symmetry vector field
on $T\VB$, which in general does not even split into a direct sum
of vector fields on separate vector bundles \cite{Kurz:2011aa}. Such
general vector fields, however, would correspond to non-local physical
transformations where a shift of a field at a point $x$ depends on
the field value at the point $y$. Only little is known about such
vector fields and the corresponding symmetry transformations, and
their relation to physics (in terms both of applicability and of meaning)
is not clear. In the following we will therefore restrict ourselves
to the case of diagonal symmetries, and present the relation between
equations of motion and the symmetry condition in such cases. As we
will see, even in the case of diagonal vector fields this relation
is not trivial.

An additional remark is in order. Due to the mentioned problem of
non-equivalent Lagrangians that differ by a total divergence on $\BM$,
the extension of the geometrical symmetry group to divergence symmetries
can not be carried over from the local theory of variational problems.
This extension is very important for a complete characterization of
(local) symmetries of non-local actions, but goes beyond the scope
of this paper (for further references see \cite{Olver:1998aa}).

\section{generalised conservation law\label{sec:Generalized-conservation-law}}

From the above treatment we see that, contrary to the local case,
the equations of motion do not directly appear in the equation encoding
the existence of a symmetry. This makes it difficult to combine these
two in order to obtain physical conservation laws. In this section
we introduce a possible connection of the two equations.

After briefly recalling the definition of the Fréchet derivative we
show that it provides a natural connection between symmetries and
equations of motion. This leads us to the generalization of the usual
conservation law, which we state in theorem \eqref{thm:Main result}.

For the convenience of the reader we recall here the equations of
motion and the symmetry equation, which we are going to combine in
what follows 
\begin{align}
EL\left(q\right)\left[X\right] & =\sum_{i}\int_{\Omega}\,\delta^{i}\left(q\right)E^{i}\left[X\right]\,\text{Vol} & \text{div}\left(\CD{\JS}L\left(X_{Q}\right)\right)+\text{div}\left(L\cdot X_{\BM}\right)+E\left(X_{Q}\right) & =0\quad.
\end{align}

We begin recalling the definition of the Fréchet derivative adopting
the notation from \cite{Olver:1998aa}. We denote a set of smooth
functions that depend on the base points, fields and their derivatives
to some finite order by $\mathcal{A}$, and further denote a space
of $l$-tuples of differentiable functions by $\mathcal{A}^{l}$,
that is functions $P=\left(P^{1},\cdots,P^{l}\right)$ where each
$P^{i}\in\mathcal{A}$. The Fréchet derivative is defined as follows. 
\begin{defn}
Let $P\in\mathcal{A}^{l}$ and $Q\in\mathcal{A}^{m}$ then the Fréchet
derivative of $P$ in the direction $Q$ is a differential operator
$D_{P}:\mathcal{A}^{m}\to\mathcal{A}^{l}$ defined as 
\begin{equation}
D_{P}\left(Q\right)=\partial_{\epsilon}\vert_{0}P\left(q,\ef+\epsilon Q\left(\ef\right),\CD{}\left(\ef+\epsilon Q\left(\ef\right)\right),\cdots\right)\quad.
\end{equation}

\end{defn}
We can now prove the following theorem. 
\begin{thm}
\label{thm:Main result} $X$ is a symmetry vector field of the non
local action $S$ with the characteristic $Q$ if and only if 
\begin{equation}
EL\left(X_{Q}\right)\left(z\right)+\sum_{i=1}^{n}\text{div}_{\dom i}\left(A^{i}\right)\left(z\right)+\int_{\Omega}\text{div}_{\BM}\left(B\right)\delta^{\alpha}\left(z\right)+\int_{\Omega}\sum_{i=1}^{n}D_{L}\left(Q^{i}\right)\left[\delta^{\alpha}-\delta^{i}\right]\left(z\right)=0\quad,\label{eq:Equation from the main result}
\end{equation}
for all $\alpha\in\left\{ 1,\cdots,n\right\} $, with $A$ being a
$n$-tuple 
\begin{equation}
A^{i}\left(z\right)=\int_{\Omega}D_{J^{i}}L\left(X^{i}\right)\delta^{i}\left(z\right)\quad,
\end{equation}
and $B\in TM$ as 
\begin{equation}
B=L\cdot X\quad.
\end{equation}

\end{thm}
Before proving this theorem we first show the following trivial identity. 
\begin{lem}
\label{lem:The-symmetry-condition in Frechet style} The symmetry
condition from the corollary \eqref{cor:Local symmetry with Euler one form}
reads in terms of Fréchet derivatives of the Lagrangian as 
\begin{equation}
D_{L}\left(Q^{1}\right)+\cdots+D_{L}\left(Q^{n}\right)=-\text{div}\left(L\cdot X_{M}\right)\quad.\label{eq:Symmetry in reduced Frechet}
\end{equation}
where $X$ is the diagonal symmetry vector field, $Q^{i}$ the characteristic
function of each individual component.\end{lem}
\begin{proof}
Let $Q$ denote the characteristic function of the symmetry vector
field $X$. The Fréchet derivative of the Lagrangian in the direction
of $Q$ is 
\begin{equation}
D_{L}\left(Q\right)=\text{div}\left(D_{J}L\left(X_{Q}\right)\right)+E\left(X_{Q}\right)\quad.
\end{equation}
Corollary \eqref{cor:Local symmetry with Euler one form} implies
\begin{equation}
D_{L}\left(Q\right)+\text{div}\left(L\cdot X_{M}\right)=0\quad.\label{eq:Frechet symmetry}
\end{equation}
On the other hand, since the symmetry vector field is diagonal, we
know that its characteristic vector field can be fibered as $X_{Q}=X_{Q}^{1}+\cdots+X_{Q}^{n}$
where $X_{Q}^{i}$ are the characteristic vector fields on the $i$th
factor of $\VB$. Moreover they are all equal. Denote with $\tilde{Q}$
the characteristic of the single field $X_{Q}^{i}$ and write $Q=Q^{1}+Q^{2}+\cdots+Q^{n}$
with $Q^{i}=\left(0,\cdots,\tilde{Q},\cdots0\right)=\tilde{Q}\otimes e_{i}$.
We then get 
\begin{equation}
D_{L}\left(Q\right)=D_{L}\left(Q^{1}\right)+\cdots+D_{L}\left(Q^{n}\right)\quad.\label{eq:Split of Frechet}
\end{equation}
Each $D_{L}\left(Q^{i}\right)$ is of the above form, namely 
\begin{equation}
D_{L}\left(Q^{i}\right)=\text{div}_{\dom i}\left(\CD{_{\JS}}L\left(X_{Q}^{i}\right)\right)+E\left(X_{Q}^{i}\right)\quad.\label{eq:Single Frechet derivative}
\end{equation}
Inserting Eq.\eqref{eq:Split of Frechet} in \eqref{eq:Frechet symmetry}
proves the lemma. 
\end{proof}
We can now prove the main theorem. 
\begin{proof}
Convoluting equation \eqref{eq:Single Frechet derivative} with $\delta^{i}\left(z\right)$
we get 
\begin{equation}
\int_{\Omega}D_{L}\left(Q^{i}\right)\delta^{i}\left(z\right)=\int_{\Omega}\text{div}\left(\CD{\JS}L\left(X_{Q}^{i}\right)\right)\delta^{i}\left(z\right)+\int_{\Omega}E\left(X_{Q}^{i}\right)\delta^{i}\left(z\right)\quad.\label{eq:Frechet and euler lagrange}
\end{equation}
Due to Eq. \eqref{eq:Variation of the action in almost a symmetry}\eqref{eq:Fibration of the divergence}
the only non vanishing divergence is on the subspace $\dom i\subset\BM$
on which the integrand is convoluted with $\delta^{i}\left(z\right)$
for $z\in M_{B}$. We therefore get 
\begin{equation}
\int_{\Omega}\text{div}\left(\cd{\JS}L\left(X_{Q}^{i}\right)\right)\delta^{i}\left(z\right)=\text{div}_{\dom i}\left(A^{i}\left(z\right)\right)\quad.
\end{equation}
The second term of equation \eqref{eq:Frechet and euler lagrange}
is a component of the Euler Lagrange equation applied on the vector
field $X_{Q}^{i}$. By successive adding of equation \eqref{eq:Frechet and euler lagrange}
for all $i$'s we obtain 
\begin{equation}
\sum_{i=1}^{N}\int_{\Omega}D_{L}\left(Q^{i}\right)\delta^{i}\left(z\right)=\sum_{i=1}^{N}\text{div}_{\dom i}\left(A^{i}\left(z\right)\right)+EL\left(X_{Q}\right)\quad.
\end{equation}
By lemma \eqref{lem:The-symmetry-condition in Frechet style} we replace
the left hand side by 
\begin{equation}
\sum_{i=1}^{N}\int_{\Omega}D_{L}\left(Q^{i}\right)\left[\delta^{i}\left(z\right)-\delta^{\alpha}\left(z\right)\right]-\int_{\Omega}\text{div}\left(L\cdot X_{\BM}\right)\delta^{\alpha}\left(z\right)=\sum_{i=1}^{N}\text{div}_{\dom i}\left(A^{i}\right)+EL\left(X_{Q}\right)\quad,
\end{equation}
where the term $\sum_{i=1}^{N}\text{div}_{\dom i}\left(A^{i}\right)$
is a sum of divergences on each individual domain $\dom i$ .

Conversely, if equation \eqref{eq:Equation from the main result}
holds then essentially following the previous steps backwards we arrive
at 
\begin{equation}
\partial_{\epsilon}\int_{c_{\BM\epsilon}\left(\Omega\right)}L\left(j\left(\Phi_{\epsilon}^{\alpha}\right)\right)\,\delta^{\alpha}\left(q\right)\,\VOL=0\quad,
\end{equation}
where $\Phi_{\epsilon}^{\alpha}$ is the section in which only the
$\alpha$'s field is transformed by $c_{\FB}\circ\phi^{\alpha}\circ c_{\BM}^{-1}$
and the rest $N-1$ fields are unchanged. By integration we get 
\begin{equation}
\partial_{\epsilon}\int_{c_{\BM\epsilon}\left(\Omega\right)}L\left(j\left(\Phi_{\epsilon}^{\alpha}\right)\right)\,\VOL=0\quad.
\end{equation}
Since Eq. \eqref{eq:Equation from the main result} holds for all
$\alpha$ we get, for $c_{\epsilon}$ close to the identity, 
\begin{equation}
0=\sum_{\alpha=1}^{N}\partial_{\epsilon}\int_{c_{\BM\epsilon}\left(\Omega\right)}L\left(j\left(\Phi_{\epsilon}^{\alpha}\right)\right)\,\VOL=\partial_{\epsilon}\int_{c_{\BM\epsilon}\left(\Omega\right)}L\left(j\left(\Phi_{\epsilon}\right)\right)\,\VOL=\partial_{\epsilon}\tilde{S}\left(\ef_{\epsilon}\right)\quad.
\end{equation}
The usual group argument concludes the proof. 
\end{proof}
We define $\Delta^{\alpha}\left(z\right):=\sum_{i}\int_{\Omega}D_{L}\left(Q^{i}\right)\left[\delta^{i}-\delta^{\alpha}\right]$$\left(z\right)$
and refer to it as to the \textit{non-local correction term} (to the
standard local conservation law). Theorem \eqref{thm:Main result}
states that the correction term is a divergence up to boundary terms
\begin{equation}
\Delta^{\alpha}\left(z\right)-\text{div}\left(A\right)\left(z\right)-\int_{\Omega}\text{div}\left(L\cdot X_{M}\right)\delta^{\alpha}\left(z\right)=0\quad.
\end{equation}

\subsection*{Action with multiple kernels\label{sec:action-with-multiple}}

So far we have discussed the case in which the action is given by
a single non-local Lagrangian. The usual action of non-local field
theories is the combination of two or more terms with different type
and combinatorics of non-localities. For instance the kinetic part
is typically assumed to be local, whereas the interaction part can
have non-localities and the various interaction terms have different
types of them. As an example consider again a non-relativistic field
theory of some atomic system, characterized by both a 2-body and a
3-body electromagnetic interaction with a Coulomb potential.

We can treat these cases by generalizing slightly our definition of
the non-local action. 
\begin{defn}
A non-local physical action with multiple kernels $S^{mult}\left[\phi\right]$
is a functional on $\Gamma\left(E\right)$, given by $\mult$ geometrical
actions $S_{l}^{G}$, each of which is defined on the product bundle
$\VB_{l}=f_{l}^{*}E^{\times l}$, for $l\in\left\{ 1,\cdots,\mult\right\} $.
\begin{equation}
S^{mult}\left[\phi\right]=\sum_{l=1}^{\mult}S_{l}^{G}\left[\imath_{l}\phi\right]\quad.
\end{equation}

\end{defn}
It is a straightforward calculation to check that the equations of
motion are now given by 
\begin{align}
\sum_{l}\,EL_{l}\left(q\right)\left[\cdot\right] & =0\quad,\label{eq:Multi kernel EoM}
\end{align}
where $EL_{l}$ is the Euler Lagrange form for the $l$th Lagrangian.

For the treatment of symmetries we again focus on the geometrical
side, which is now the sum of individual functionals, each of which
involve different domains of integration. It becomes difficult to
formulate the most general geometrical transformation, since the tangent
space on $\VB_{l}$ is different from $\VB_{m}$ if $l\neq m$. However,
we can make sense of a symmetry transformation, if we again restrict
ourselves only to diagonal vector fields, $X^{l}=\left(X,\cdots,X\right)$
with $X\in TE$. In this case we get 
\begin{lem}
$X\in TE$ is a symmetry vector field of $\sum_{l=1}^{\mult}S_{l}^{G}$
iff it is a symmetry of each action $S_{l}^{G}$ separately. That
is $\text{d}S_{l}^{G}\left(X^{l}\right)=0$ for each $l\in\left\{ 1,\cdots,\mult\right\} $. \end{lem}
\begin{proof}
Let $c_{E\epsilon}^{l}$ and $c_{M\epsilon}^{l}$ be the flows of
$X^{l}$. The variation of the generalised geometrical action is 
\begin{equation}
\partial_{\epsilon=0}\sum_{l=1}^{\mult}\tilde{S}_{l,\epsilon}^{G}\left(\phi_{\epsilon}\right)=0\quad.
\end{equation}
Each action contains an independent domain of integration. Since we
defined the symmetry as a transformation that does not change the
action for all subdomains, each term has to vanish separately. This
can be seen in the following way. For $i\in\left\{ 1,\cdots,\mult\right\} $
the variation of the $i$th action has to satisfy 
\begin{equation}
\text{d}S_{\Omega_{i}}\left(X^{i}\right)=-\sum_{l\neq i}^{\mult}\text{d}S_{\Omega_{l}}\left(X^{l}\right)\quad.
\end{equation}
Fixing $\Omega_{l\neq i}$, we get 
\begin{equation}
\text{d}S_{\Omega_{i}}\left(X^{i}\right)=\text{const}\quad\forall\tilde{\Omega}_{i}\subset\Omega_{i}\quad.
\end{equation}
Choosing $\Omega''\subset\Omega'\subset\Omega_{i}$ and using the
definition of the action we get 
\begin{align}
0 & =\text{d}S_{\Omega'}\left(X^{i}\right)-\text{d}S_{\Omega''}\left(X^{i}\right)=\text{d}S_{\Omega''}\left(X^{i}\right)+\text{d}S_{\Omega'\backslash\Omega''}\left(X^{i}\right)-\text{d}S_{\Omega''}\left(X^{i}\right)\nonumber \\
 & =\text{d}S_{\Omega'\backslash\Omega''}\left(X^{i}\right)\quad.
\end{align}
Taking $\tilde{\Omega}\subset\Omega'\backslash\Omega''\subset\Omega_{i}$
leads to 
\begin{equation}
\text{d}S_{\tilde{\Omega}}\left(X^{i}\right)=0\quad,
\end{equation}
which determines the constant. Applying the same reasoning for the
remaining part $\sum_{l\neq i}^{\mult}\text{d}S_{\Omega_{l}}\left(X^{l}\right)=0$
proves one direction of the lemma. The opposite direction is trivial. 
\end{proof}
For each action we get a symmetry equation of the form, 
\begin{align}
EL_{l}\left[X_{Q}\right]\left(q\right)+\sum_{i=1}^{N}\text{div}_{\dom i}\left(A^{i,l}\right)\left(q\right)_{l}+\int_{\Omega}\text{div}_{\BM}\left(B^{l}\right)\delta^{\alpha}\left(q\right)_{l}-\Delta^{\alpha,l}\left(q\right) & =0\quad.
\end{align}
The explicit form of $A^{i.l},\,B^{l}$ and $\Delta^{\alpha,l}$ is
given in the following theorem. Summing over $l$ we obtain the generalised
conservation laws for the multiple action which leads us to our final
result. 
\begin{thm}
Let $S$ be the non-local physical action given by a sum of geometrical
actions as $\sum_{l=1}^{\mult}S_{l}^{G}$ . Let $X$$^{l}=X^{\times l}$
be the symmetry of the $l$th geometrical action. The generalised
conservation law for this symmetry takes the form 
\begin{equation}
EL\left[X_{Q}\right]\left(q\right)+\sum_{l=1}^{\mult}\left\{ \sum_{i=1}^{N}\text{div}_{\dom i}\left(A^{i,l}\right)\left(q\right)_{l}+\int_{\Omega}\text{div}_{\BM}\left(B^{l}\right)\delta^{\alpha}\left(q\right)_{l}-\Delta^{\alpha,l}\left(q\right)\right\} =0\quad,
\end{equation}
with $A^{i,l}$ being 
\begin{equation}
A^{i,l}\left(q\right)=\int_{\Omega}\CD{\JS^{i}}L^{l}\left[X^{i}\right]\delta^{i}\left(z\right)\,\VOL^{l}\quad,
\end{equation}
$B^{l}\in TM$ as 
\begin{equation}
B=L^{l}\cdot X^{l}\quad,
\end{equation}
and the non-local correction term 
\begin{align}
\Delta^{\alpha,l}\left(z\right) & =\sum_{i=1}^{l}\int_{\Omega}D_{L}\left(Q^{i}\right)\left[\delta^{i}-\delta^{\alpha}\right]\left(z\right)\,\VOL^{l}\quad.
\end{align}

\end{thm}
For physical reasons we separate a special case as a corollary. 
\begin{cor}
Let the non-local physical action be composed of a local and a single
non-local part, which we call the interaction. And let $X^{\text{loc}}$
and $X^{\text{nloc}}=X^{\text{loc}\times N}$ be the symmetries of
the local and the non-local geometrical action, respectively. The
corresponding generalised conservation law reads 
\begin{equation}
\text{div}_{M}\left(A^{\text{loc}}+B^{\text{loc}}\right)\left(q\right)+\sum_{i=1}^{N}\text{div}_{\dom i}\left(A^{i,\text{int}}\right)\left(q\right)+\int_{\Omega}\text{div}_{\BM}\left(B^{\text{int}}\right)\delta^{\alpha}\left(q\right)-\Delta^{\alpha,\text{int}}\left(q\right)=0\quad,
\end{equation}
with $A$, $B$ and $\Delta$ as above. 
\end{cor}
One sees that in this case the non-local correction originates only
from the non-local interaction term. We will explicitly apply this
result in the next section.

\subsection*{On conserved charges}

In local field theories in which the base manifold is associated with
space-time itself, one can integrate the conservation law over the
spatial degrees of freedom. The resulting quantity is then a function
of time only, the \lq charge\rq. It is a consequence of the local
conservation law that the charge is conserved, meaning that its time
derivative vanishes. In the Hamiltonian language charges become the
generators of the symmetry transformation, and might completely determine
the dynamics of the system.

In more abstract field theories, such as group field theories, where
the base manifold has no a priori relation to space-time, the notion
of \lq conserved\rq~ charges has no obvious counterpart.  Nevertheless,
in models in which the combinatorial structure of non-localities is
such that the theory is local in one of the parameters (for example
as it is the case for the time submanifold in non-relativistic field
theories) the integral over the correction term $\Delta$ vanishes,
as it can be easily verified, and the corresponding generalised charge
will be conserved with respect to the local parameter. However, in
this generalized sense it is not clear if the charge will correspond
to the canonical structure of the theory, since due to the absence
of time the canonical structure is not well defined. 

The analysis of such generalised notion of conserved charges needs
to be carried out in greater detail than we are doing here, as it
has important consequences. In space-time-based non-local field theories
the notion of charges can help to define the Hamiltonian structure
in the absence of a unique notion of time \cite{Pauli:1953aa}, but
it may also simplify the equations of motion by reducing the degree
of differential operators, as it does in local theories (see \cite{Olver:1998aa}
for local examples).

In more abstract field theories the notion of charges will in general
be very different from ordinary field theory due to the lack of direct
physical interpretation of the background manifold. However, it can
lead to useful insights both from the fundamental as well as from
the computational point of view.

\section{Physical examples\label{sec:Physical-examples}}

We now apply our results to a non-local model of a complex scalar
field with a non-local two-body interaction. Such model often appears
in many-body quantum systems such as hydrodynamics, the theory of
Bose-Einstein condensates or solid state physics. In this section
we first introduce the action and the geometrical space according
to our definitions. We then consider two physical symmetries and calculate
the corresponding currents, together with their generalised conservation
laws. Subsequently we show that for a particular choice of the interaction,
the action can be rewritten as a local action, for which the usual
Noether theorem holds. We conclude the section by comparing the currents
of the local and non-local theories explicitly showing the appearance
of the non-local correction term.

Our main purpose here is not to analyze the corresponding physics
or explore the physical consequences of symmetries and their conservations
laws for such systems. It is only to illustrate in some detail, on
one specific example of a clear physical significance, the general
formalism we have developed for carrying out such analysis in generic
non-local field theories.

\subsection{Non-local theory }

We consider a model given by a complex scalar field $\phi\left(t,x\right)$
and the non-local physical action 
\begin{align}
S & =\int\text{d}t\text{d}x\,\imath\phi^{*}(t,x)\partial_{t}\phi(t,x)+\nabla\phi^{*}(t,x)\cdot\nabla\phi(t,x)\nonumber \\
 & +\int\text{d}t\text{d}x\text{d}y\,V(x,y)\vert\phi(t,x)\vert^{2}\ \vert\phi(t,y)\vert^{2}\quad.\label{eq:Non local action}
\end{align}
We can think of this system as describing the hydrodynamic approximation
of an atomic fluid. The local part of this model is the usual non-relativistic
kinetic term. The non-local part describes the instantaneous interaction
of two atoms at distinct space points. The nature of the interaction
is determined by the potential term $V\left(x,y\right)$, which, for
now, we leave unspecified except for the assumption that it vanishes
sufficiently fast at infinity. The fields in such models are assumed
to be normalized to the number of atoms in the location $x$. The
density of such a continuous system is defined by the modulus square
of the field $\rho\left(t,x\right)=|\phi\left(t,x\right)|^{2}$.

The theory is local in time and non-local in space. The geometrical
action is of the multi kernel type with the local and interaction
part given by 
\begin{align}
S^{1} & =\int\text{d}t\text{d}x\,\imath\phi^{*}(t,x)\partial_{t}\phi(t,x)+\nabla\phi^{*}(t,x)\cdot\nabla\phi(t,x)\\
S^{2} & =\int\int\text{d}t\text{d}x\text{d}y\,V(x,y)\vert\phi(t,y)\vert^{2}\ \vert\phi(t,y)\vert^{2}\quad.
\end{align}
The bundle structure of $\VB$ can be represented as in Figure \eqref{fig:Bundle-structure.-Lower},
where the solid line denotes the base manifold and the dashed line
refer to the fibers.

\begin{figure}
\begin{centering}
\begin{tikzpicture}

\begin{scope}[xshift=-5cm]
\draw[line width=1pt] (0,0) -- node [above] {$(\mathbb{R},t)$} ++(1,0);
\draw[line width=1pt] (1.4,0) -- node [above] {$(\mathbb{R}^3,x)$} ++ (1.1,0);
\draw[dashed] (1.2,0.2) -- (1.2,1.3) node [right] {$V$};
\draw (-0.2,-0.2) -- (-0.2,1.5) -- (2.8,1.5) -- (2.8,-0.2) -- cycle;
\node at (1.2,0) {$\times$};
\node at (0,1.3) {$E$};
\end{scope}

\begin{scope}[xshift=-5cm, yshift=2cm]
\draw[line width=1pt] (0,0) -- node [above] {$(\mathbb{R},t)$} ++(1,0);
\draw[line width=1pt] (1.4,0) -- node [above] {$(\mathbb{R}^3,y)$} ++ (1.3,0);
\draw[dashed] (1.2,0.2) -- (1.2,1.3) node [right] {$V$};
\draw (-0.2,-0.2) -- (-0.2,1.5) -- (2.8,1.5) -- (2.8,-0.2) -- cycle;
\node at (1.2,0) {$\times$};
\node at (0,1.3) {$E$};
\end{scope}

\begin{scope} [xshift=1cm, yshift=1cm]
\draw[line width=1pt] (0,0) -- node [below] {$(\mathbb{R}^3,x)$} (1.3,0);
\draw[line width=1pt] (1.7,0) -- node [below] {$(\mathbb{R},t)$} (2.6,0);
\draw[line width=1pt] (3,0) -- node [below] {$(\mathbb{R}^3,y)$} (4.3,0);
\draw[dashed] (1.5,0.3) --  (1.5,1.5);
\draw[dashed] (2.8,0.3) --  (2.8,1.5);

\draw (-0.2,-0.6) -- (-0.2,1.8) -- (4.5,1.8) -- (4.5,-0.6) -- cycle;
\node at (1.5,0) {$\times$};
\node at (2.8,0) {$\times$};
\node at (0,1.5) {$\mathcal{E}$};
\end{scope}

\draw[->,xshift=-1.7cm,>=stealth] (0,1.6) -- node [above] {$f^{*}$} (2,1.6);

\end{tikzpicture} 
\par\end{centering}

\protect\protect\caption{Bundle structure. Left hand side is the bundle with trivial combinatorics.
The right hand side denote the combinatorics of the pull back bundle
with the embedding map $f$ as in Eq. \eqref{eq:an example for f}.
Lower case letters denote the coordinate names. \label{fig:Bundle-structure.-Lower}}
\end{figure}
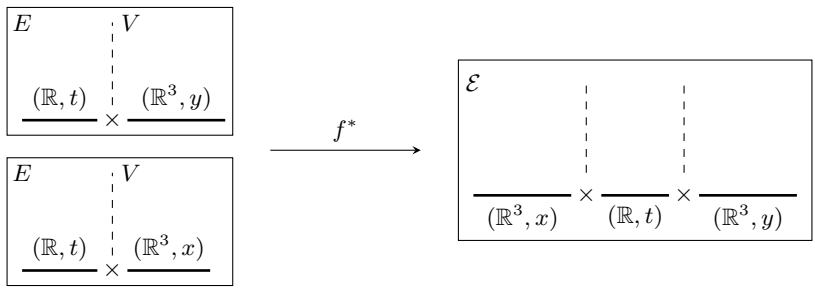

Explicitly the geometric structure of this example is given by the
pull back of two copies of the vector bundle $\vb=M\times\mathbb{C}$
with the combinatorial function $f$. The base manifold $M$ in this
example is the product manifold of time and space coordinates $\bm=\mathbb{R}\times\mathbb{R}^{3}$.
The fiber is a vector space of complex numbers $\mathbb{C}$. And
the combinatorial structure of the theory is given by the following
function 
\begin{equation}
f=\left(\mathds{1}\times\text{Per}\times\mathds{1}\right)\left(\text{Di}\times\mathds{1}\times\mathds{1}\right)\quad.
\end{equation}
In this way the full jet bundle writes as
\begin{equation}
\mathcal{J}\left(E^{P}\right)\simeq\overbrace{\underbrace{\overbrace{\mathbb{R}\times\mathbb{\mathbb{R}}^{3}\times\mathbb{R}^{3}}^{\text{Base manifold}\,\BM}}_{\left(t,x,y\right)}\times\overbrace{\underbrace{\mathbb{C}}_{\left(u\right)}\times\underbrace{\mathbb{C}}_{\left(v\right)}}^{\text{Fibers}\,\FB}}^{\VB}\times\overbrace{\underbrace{\mathbb{C}\times\mathbb{C}^{3}}_{\left(u_{t},u_{x}\right)}\times\underbrace{\mathbb{C}\times\mathbb{C}^{3}}_{\left(v_{t},v_{x}\right)}}^{\text{Jet spaces}}\quad,
\end{equation}
where the under-brace denote the coordinate convention that we use
for each of the spaces.

In these coordinates the geometrical Lagrangians of $S^{1}$ and $S^{2}$
become 
\begin{align}
L^{\text{loc}} & =\imath u^{*}u_{t}+u_{x}^{*}u_{x}\\
L^{\text{nloc}} & =V\left(x,y\right)|u|^{2}\,|v|^{2}\quad.
\end{align}

\subsubsection{$U(1)$ symmetry}

We now consider the symmetry of the action under the phase rotation
of the fields $\phi$. The geometrical action is symmetric under the
transformation $c_{\fb\epsilon}\left(v\right)=e^{\imath\epsilon\theta}v$
and $c_{\bm\epsilon}=\mathds{1}$. Its infinitesimal generator is
given by 
\begin{align}
X & =\imath\alpha v\,\partial_{v}-\imath\alpha v^{*}\partial_{v^{*}}\quad,
\end{align}
which is already in its characteristic form. The characteristic function
$Q$ is therefore $Q=\left(\imath\alpha v,-\imath\alpha v^{*}\right)$.
Using this it is immediate to notice that the non-local correction
term $\Delta^{1}$ to the conservation law vanishes 
\begin{align}
\Delta^{1} & =\partial_{\epsilon}V\left(x,y\right)|u|^{2}|v+\epsilon Q\left(v\right)|^{2}=V\left(x,y\right)|u|^{2}|v\vert^{2}\left(\imath\alpha-\imath\alpha\right)=0\quad.
\end{align}
We therefore obtain the continuity equation 
\begin{equation}
\text{div}\left(J^{\text{loc}}\right)+\text{div}\left(J^{\text{nloc}}\right)=0\quad,
\end{equation}
where a priori one has an additional non-local contribution to the
standard local one. However, since $D_{J}L^{\text{nloc}}=0$ as well
as $X_{M}=0$, the non-local currents vanish, $J^{\text{nloc}}=0$,
leading to the standard conservation law 
\begin{equation}
EL\left(X_{Q}\right)+\text{div}\left(J^{\text{loc}}\right)=0\quad,
\end{equation}
with the well-known conserved current 
\begin{align}
J^{\text{loc}} & =-\vert u\vert^{2}\,\partial_{t}+\imath\left(u_{x}^{*}u-u^{*}u_{x}\right)\partial_{x}\quad.
\end{align}
Alternatively inserting the coordinates $u=\phi(t,x)$ the continuity
equation on-shell gets the usual form 
\begin{align}
\text{div}\left(J^{\text{loc}}\right)=-\partial_{t}\big(\phi^{*}\phi\big)+\imath\,\text{div}_{x}\big(\phi\nabla\phi^{*}-\phi^{*}\nabla\phi\big) & =0\quad,
\end{align}
where $\nabla$ is the spatial gradient.

\subsubsection{Translation invariance}

If we assume that the potential $V$ depends only on the distance
between the points $V(x,y)=V(\vert x-y\vert)$, the action also becomes
translation invariant. That means that the geometrical action is invariant
under a symmetry transformation with a vector field of the following
form 
\begin{align}
X & =\epsilon^{t}\partial_{t}+\epsilon^{1}\partial_{x^{1}}+\epsilon^{2}\partial_{x^{2}}+\epsilon^{3}\partial_{x^{3}}=:\epsilon^{t}\partial_{t}+\vec{\epsilon}\cdot\partial_{x}\quad.
\end{align}
The characteristic of the vector field has now two components 
\begin{align}
Q^{2} & =-(u_{t}\epsilon^{t}+u_{x}\cdot\vec{\epsilon}) & Q^{2*} & =-(v_{t}^{*}\epsilon^{t}+v_{x}^{*}\cdot\vec{\epsilon})\quad,
\end{align}
where we denote the scalar product between the 3-dimensional jet coordinate
$u_{x}$ and $\vec{\epsilon}\in\mathbb{R}^{3}$ with the dot. The
Fréchet derivative is 
\begin{align}
D_{L}\left(Q^{2}\right) & =-V\left(x-y\right)|u|^{2}\left(\epsilon^{t}\vert v_{t}\vert^{2}+\vec{\epsilon}\cdot\vert v_{x}\vert^{2}\right)\quad.
\end{align}
The non-local correction term reads 
\begin{align}
\Delta^{1}\left(t,z\right) & =-\int_{\mathbb{R}^{3}\times\mathbb{R}^{3}}V\left(x-y\right)|u|^{2}\left(\epsilon^{t}\vert v_{t}\vert^{2}+\vec{\epsilon}\cdot\vert v_{x}\vert^{2}\right)\left[\delta\left(z-y\right)-\delta\left(z-x\right)\right]\quad.
\end{align}
By Corollary 3, it follows that 
\begin{align}
\text{div}_{t,z}\left(T^{\text{loc}}\cdot\epsilon\right) & +\int\text{d}x\,\text{div}_{t,z,x}\left(L^{\text{nloc}}\left(t,z,y\right)\cdot\epsilon\right)=\Delta^{1}\left(t,z\right)\quad.
\end{align}
Observing the $t,z$ divergence in the energy momentum tensor, we
get 
\begin{equation}
\text{div}_{t,z}\left(T\cdot\epsilon\right)=\Delta^{1}\left(t,z\right)-\int\text{d}x\,\text{div}_{x}\left(L^{\text{nloc}}\left(t,z,y\right)\cdot\epsilon^{x}\right)\quad,
\end{equation}
with the energy-momentum tensor of the non-local theory given by 
\begin{equation}
T_{\mu\nu}\left(z\right)=T_{\mu\nu}^{\text{loc}}\left(z\right)+\int_{\mathbb{R}^{3}}\text{d}y\,V\left(z-y\right)\rho\left(t,z\right)\rho\left(t,y\right)\delta_{\mu\nu}\quad.\label{eq:Non local energy momentum}
\end{equation}
By partial integration and again using the coordinates $u=\phi\left(t,x\right)$,
we get

\begin{align}
\text{div}\left(T\cdot\epsilon\right) & =\int\text{d}x\,\left(\epsilon\cdot\nabla_{z}\right)V\left(z-x\right)\rho\left(t,z\right)\rho\left(t,x\right)-V\left(x-z\right)\rho\left(t,x\right)\left(\epsilon\cdot\nabla_{z}\right)\rho\left(t,z\right)\nonumber \\
 & +\int\text{d}x\,V\left(z-x\right)\rho\left(t,z\right)\left(\epsilon_{t}\cdot\partial_{t}\right)\rho\left(t,x\right)-V\left(x-z\right)\rho\left(t,x\right)\left(\epsilon_{t}\cdot\partial_{t}\right)\rho\left(t,z\right)\quad.\label{eq:Energy momentum with correction term}
\end{align}

Note that, in the last equation, no boundary terms are present since
they canceled by moving the gradient of the first term to the potential.
Notice also that, if the potential $V$ is a contact interaction,
that is if $V\left(x-y\right)=\delta\left(x-y\right)$, the interaction
term becomes local and the terms on the right-hand-side of Eq. \eqref{eq:Non local energy momentum}
cancel each other, leading to a conserved energy-momentum tensor.

Defining the mean value as $\langle\circ\rangle_{t}=\int\text{d}x\,\circ\,\rho\left(t,x\right)$
, we can rewrite equation \eqref{eq:Energy momentum with correction term}
as 
\begin{equation}
\text{div}\left(T\cdot\epsilon\right)=\rho\left(t,z\right)\epsilon\cdot\nabla_{t,z}\langle V\rangle\left(t,z\right)-\epsilon\cdot\nabla_{t,z}\rho\left(t,z\right)\langle V\rangle\left(t,z\right)\quad,
\end{equation}
which admits a nice physical interpretation. The divergence of the
energy-momentum tensor is given by the mean force $\nabla_{t,x}\langle V\rangle$
due to non-local interactions and the dilution force $-\nabla_{t,z}\rho$
due to density fluctuations. Assuming a contact interaction, $V\left(x,y\right)=\delta\left(x-y\right)$,
we obtain that the right-hand-side of the above equation cancels.
In this case the mean force at a point has the same strength but the
opposite direction as the dilution force. In the non-local case on
the other hand the balance is broken since the fields at a distance
induce an external mean field at each point turning the physical situation
into that of a system in an external potential.

\subsection{Additional note: underlying local theory}

In this specific example, the non-local theory can be understood as
an approximation of a local theory with additional degrees of freedom.
We use this fact to give an interpretation for the additional terms
in our generalised conservation law, as compared to the usual one.
In fact, this is nothing else than a simpler case of the step from
a theory of atoms interacting via a Coulomb potential to a description
of the same atoms interacting via an electromagnetic field. Let us
point out, however, that, on the one hand, we do not know if such
rewriting of non-local field theories in terms of local ones is always
possible, and that, on the other hand, our results concerning non-local
theories do not rely on this possibility in any way.

We now present the local theory from which the above non-local emerges
and show how the local, conserved currents relate to the non-local
correction term. The Lagrangian of the local theory is given by 
\begin{equation}
L=\imath\phi^{*}\partial_{t}\phi+\nabla\phi^{*}\cdot\nabla\phi+\alpha\vert\phi\vert^{2}A+\beta\nabla A\cdot\nabla A\quad,\label{eq:Local lagrangian}
\end{equation}
with an auxiliary real scalar field $A$. The equation of motion for
$A$ gives 
\begin{equation}
\Delta_{x}A\left(x,t\right)=\frac{\alpha}{2\beta}\vert\phi\vert^{2}\left(x,t\right)\quad,
\end{equation}
which can be solved via the Green function 
\begin{align}
G\left(x,y\right) & =-\frac{1}{4\pi\vert x-y\vert}\quad.
\end{align}
We get 
\begin{equation}
A\left(x,t\right)=\frac{\alpha}{2\beta}\int_{\mathbb{R}^{3}}\text{d}y\,\rho\left(t,y\right)G\left(y,x\right)\quad.\label{eq:Solution to local A}
\end{equation}
Inserting the solution into the Lagrangian Eq.\eqref{eq:Local lagrangian},
we get 
\begin{equation}
L=L^{\text{loc}}+\frac{\alpha^{2}}{2\beta}\int_{\mathbb{R}^{3}}\text{d}y\,\rho\left(t,x\right)\rho\left(t,y\right)G\left(x,y\right)+\beta\cdot\text{div}\left(A\nabla A\right)\quad.\label{eq:equivalence of lagrangians}
\end{equation}
Calling $V\left(x-y\right):=\frac{\alpha^{2}}{2\beta}G\left(x,y\right),$
we get a non-local action from the form of Eq. \eqref{eq:Non local action}
plus a total divergence, which does not change the equations of motion.
The two actions are therefore equivalent.

Since the local and the non-local action are equivalent, we can calculate
the currents of the local model and compare them to the non-local
case.

The $U\left(1\right)$ current of the local action is the same as
for the non-local case since the real field $A$ does not change under
the transformation $\phi\mapsto\phi_{\epsilon}=c_{E\epsilon}\circ\phi$.
The energy-momentum tensor, however, is now conserved but depends
on $A$. 
\begin{equation}
\text{div}\left(T^{\phi}\cdot\epsilon+T^{A}\cdot\epsilon\right)=0\quad.
\end{equation}
The $A$-dependent part of the energy-momentum tensor reads 
\begin{equation}
\partial_{\mu}T_{\mu\nu}^{A}\,\epsilon^{\nu}=-\beta\partial_{i}\left(\partial^{i}A\,\partial_{\nu}A\right)\cdot\epsilon^{\nu}+\partial_{\mu}\left(L^{A}\cdot\delta_{\mu\nu}\right)\cdot\epsilon^{\nu}\quad,\label{eq:Local energy momentum}
\end{equation}
where we have explicitly written the divergence in components, using
greek letters to range over $\left(t,x^{1},x^{2},x^{3}\right)$ and
latin indices to range only over the spatial part.

Using the notation $\nu=\left(t,j\right)$, we calculate the first
term $\partial_{i}\left(\partial^{i}A\,\partial_{\nu}A\right)$ as
\begin{align}
\partial_{i}\left(\partial^{i}A\,\partial_{\nu}A\right) & =\Delta A\,\partial_{t}A-A\,\partial_{t}\Delta A+\Delta A\,\partial_{j}A-A\partial_{j}\Delta A+\partial^{i}\left(A\,\partial_{i}\partial_{t}A\right)+\partial^{i}\left(A\,\partial_{i}\partial_{j}A\right)\nonumber \\
 & =\frac{1}{\beta}\left[\int V\left(y,z\right)\rho\left(t,z\right)\partial_{t}\rho\left(t,y\right)\text{d}y-\int V\left(y,z\right)\partial_{t}\rho\left(t,z\right)\rho\left(t,y\right)\text{d}y\right]\nonumber \\
 & +\frac{1}{\beta}\left[\int\partial_{j}^{z}V\left(z,y\right)\rho\left(t,z\right)\rho\left(t,y\right)\text{d}y-\int V\left(y,z\right)\partial_{j}^{z}\rho\left(t,z\right)\rho\left(t,y\right)\text{d}y\right]\nonumber \\
 & +\partial_{\mu}\left(A\,\partial_{\nu}\partial^{\mu}A\right)\quad,\label{eq:Local energy momentum tensor part}
\end{align}
where the first equation follows from partial integration and the
Schwarz theorem, used to interchange the derivatives. Inserting Eq.\eqref{eq:Local energy momentum tensor part}
into Eq.\eqref{eq:Local energy momentum}, we get the divergence of
the energy-momentum tensor in the local theory 
\begin{align}
\text{div}_{t,z}\left(T^{A}\cdot\epsilon\right) & =-\left[\int V\left(x,z\right)\rho\left(t,z\right)\left(\epsilon_{t}\cdot\partial_{t}\right)\rho\left(t,x\right)\text{d}x-\int V\left(x,z\right)\left(\epsilon_{t}\cdot\partial_{t}\right)\rho\left(t,z\right)\rho\left(t,x\right)\text{d}x\right]\nonumber \\
 & -\left[\int\left(\vec{\epsilon}\cdot\nabla_{z}\right)V\left(z,x\right)\rho\left(t,z\right)\rho\left(t,x\right)\text{d}x-\int V\left(x,z\right)\left(\vec{\epsilon}\cdot\nabla_{z}\right)\rho\left(t,z\right)\rho\left(t,x\right)\text{d}x\right]\nonumber \\
 & -\beta\cdot\text{div}\left(A\,\left(\epsilon\cdot\nabla\right)\nabla A\right)\nonumber \\
 & +\partial_{\mu}\left(L^{A}\cdot\delta_{\mu\nu}\right)\cdot\epsilon^{\nu}\quad.
\end{align}
The last term can be put into the form 
\begin{equation}
\partial_{\mu}\left(L^{A}\cdot\delta_{\mu\nu}\right)\cdot\epsilon^{\nu}=\partial_{\mu}\left(\int\text{d}y\,L^{\text{nloc}}\left(z\right)\cdot\delta_{\mu\nu}\right)\cdot\epsilon^{\nu}+\partial_{\mu}\partial_{i}\left(A\partial^{i}A\cdot\delta_{\mu\nu}\epsilon^{\nu}\right)\quad,
\end{equation}
which is the $B=L^{\text{nloc}}\cdot X$ contribution to the energy
momentum tensor. Rearranging the terms of Eq. \eqref{eq:Local energy momentum}
as follows we get 
\begin{align}
\text{div}\left(T^{\phi}\cdot\epsilon+B\right) & =\left[\int\left(\vec{\epsilon}\cdot\nabla_{z}\right)V\left(z,x\right)\rho\left(t,z\right)\rho\left(t,x\right)\text{d}x-\int V\left(x,z\right)\rho\left(t,x\right)\left(\vec{\epsilon}\cdot\nabla_{z}\right)\rho\left(t,z\right)\text{d}x\right]\nonumber \\
 & +\left[\int V\left(x,z\right)\rho\left(t,z\right)\left(\epsilon_{t}\cdot\partial_{t}\right)\rho\left(t,x\right)\text{d}x-\int V\left(x,z\right)\rho\left(t,x\right)\left(\epsilon_{t}\cdot\partial_{t}\right)\rho\left(t,z\right)\text{d}x\right]\nonumber \\
 & +\beta\cdot\partial_{i}\left[\left(A\,\left(\epsilon^{\nu}\partial_{\nu}\right)\partial^{i}A\right)-\partial_{\mu}\left(A\partial^{i}A\cdot\delta_{\mu\nu}\epsilon^{\nu}\right)\right]\quad.
\end{align}
Compared to the generalised conservation law from the non-local theory,
Eq.\eqref{eq:Energy momentum with correction term}, stated here for
convenience

\begin{align}
\text{div}_{t,z}\left(T\cdot\epsilon\right) & =\int\text{d}x\,\left(\vec{\epsilon}\cdot\nabla_{z}\right)V\left(z-x\right)\rho\left(t,z\right)\rho\left(t,x\right)-V\left(x-z\right)\rho\left(t,x\right)\left(\vec{\epsilon}\cdot\nabla_{z}\right)\rho\left(t,z\right)\nonumber \\
 & +\int\text{d}x\,V\left(z-x\right)\rho\left(t,z\right)\left(\epsilon_{t}\cdot\partial_{t}\right)\rho\left(t,x\right)-V\left(x-z\right)\rho\left(t,x\right)\left(\epsilon_{t}\cdot\partial_{t}\right)\rho\left(t,z\right)\quad,
\end{align}
we see that the non-local correction comes from the auxiliary field
$A$ and differs from a total divergence by an additional term 
\begin{equation}
\beta\cdot\partial_{i}\left[\left(A\,\left(\epsilon^{\nu}\partial_{\nu}\right)\partial^{i}A\right)-\partial_{\mu}\left(A\partial^{i}A\cdot\delta_{\mu\nu}\epsilon^{\nu}\right)\right]\quad.
\end{equation}

\section{Ward identities for non-local theories}

In quantum field theory the symmetries of the classical action help
to simplify the relations between correlation functions for the quantum
fields. These simplified relations are called Ward identities. In
the following we will review (in a sketchy fashion) the derivation
of Ward identities in the functional integral formalism \cite{Nair:2004aa},
and show that the classical results we derived above for non-local
quantum field theories, imply relations between correlation functions
just as in local quantum field theories.

\
 The generating functional in Euclidean quantum field theories is
defined in terms of the classical but possibly renormalized action
$S$ as follows 
\begin{equation}
Z\left[J\right]=\int\mathcal{D}\phi\,e^{-S\left(\phi\right)+\int J\phi}\quad.
\end{equation}
Perform a coordinate transformation $\phi\mapsto\tilde{\phi}=\phi+Q$
where $Q$ is a characteristic of a symmetry vector field $X$ of
the classical action $S$. Since we integrate over all fields the
generating functional does not change and so we get 
\begin{align}
Z\left[J\right] & =\int\mathcal{D}\phi\,e^{-S\left(\phi\right)+\int J\phi}\nonumber \\
 & =\int\mathcal{D}\tilde{\phi}\,e^{-S\left(\tilde{\phi}\right)+\int J\tilde{\phi}}=\tilde{Z}\left[J\right]\quad.\label{eq:Ward identities first step}
\end{align}
Expressing $\tilde{Z}\left[J\right]$ in terms of the un-transformed
fields $\phi$ we get to the first order in $Q$ 
\begin{align}
\tilde{Z}\left[J\right] & =\int\mathcal{D}\phi\,\left[1+\text{Tr}\left(D_{V}Q\right)\right]\,e^{-S\left(\phi\right)+\int J\phi+\text{d}S\left(X\right)+\int JX_{Q}}\quad,
\end{align}
where the term $\text{Tr}\left(\CD{\FB}Q\right)$ comes from the transformation
of the functional measure. We also expanded the action to first order
as $S\left(\tilde{\phi}\right)\approx S\left(\phi\right)+\text{d}S\left(X_{Q}\right)$.
Expanding the exponential to first order of $Q$ yields 
\begin{equation}
\tilde{Z}\left[J\right]=\int\mathcal{D}\phi\,e^{-S\left(\phi\right)+\int J\phi}\left[1+\text{Tr}\left(\CD{\FB}Q\right)-\text{d}S\left(X_{Q}\right)+\int JQ+\mathcal{O}\left(Q^{2}\right)\right]=Z\left[J\right]+\delta Z_{J}\left[X_{Q}\right]+\mathcal{O}\left(Q^{2}\right)\quad,
\end{equation}
where we define 
\begin{equation}
\delta Z_{J}\left[X_{Q}\right]=\langle\text{Tr}\left(\CD{\FB}Q\right)-\text{d}S\left(X_{Q}\right)+\int JQ\rangle\quad.\label{eq:General ward type identitiy}
\end{equation}
Using the notion of a functional average 
\begin{equation}
\langle\circ\rangle=\int\mathcal{D}\phi\,\left[e^{-S\left(\phi\right)+\int J\phi}\,\circ\quad\right]\quad.
\end{equation}
Equations \eqref{eq:Ward identities first step} implies 
\begin{equation}
\delta Z_{J}\left[X_{Q}\right]=0\quad.
\end{equation}
Due to lemma \eqref{lem:The-symmetry-condition in Frechet style},
which stated $\text{d}S\left(X_{Q}\right)=-\int_{\Omega}\text{div}\left(L\cdot X_{\BM}\right)$,
we obtain 
\begin{equation}
\delta Z_{J}\left[X_{Q}\right]=\langle\text{Tr}\left(\CD{\FB}Q\right)+\int_{\Omega}\text{div}\left(L\cdot X_{\BM}\right)+\int JQ\rangle=0\quad.\label{eq:Ward identity}
\end{equation}
The \textit{anomaly} term $\text{tr}\left(\CD{\FB}Q\right)$ comes
from the transformation of the measure. If it does not vanish, the
symmetries of the classical action do not imply a symmetry of the
quantum effective action, and the symmetry is thus broken at the quantum
level.

Equation \eqref{eq:Ward identity} defines the Ward identities.

Expanding the generating functional in powers of $J$ Eq. \eqref{eq:Ward identity}
leads to the relation between coefficients of different powers of
$J$, which correspond to different correlation functions.

So far, this statement is global, by which we mean that it depends
on the whole domain $\Omega$. However, we can use our approach to
the theorem \eqref{thm:Main result} to obtain a semi-local statement.

To do this, we multiply the characteristic $Q$ of the symmetry vector
field $X_{Q}$ by an arbitrary smooth function $\eta$ that vanishes
on the boundaries of $U$. Then, we perform the variation of the action
in the direction of $X_{Q\eta}$. Since the vector field $\imath_{N}\left(X_{Q\eta}\right)$
vanishes on the boundaries of $U$ the variation leads to the integral
form of the Euler-Lagrange equations, 
\begin{equation}
\text{d}S\left(X_{Q\eta}\right)=\sum_{i=1}^{N}\int_{\Omega}\,E_{i}\left[X_{Q}\right]\eta^{i}\,\VOL\quad,
\end{equation}
where $\eta^{i}$ denotes the smooth function $\eta$ defined over
the $M^{i}=\text{pr}^{i}\left(f\left(\BM\right)\right)$. Using Eq.\eqref{eq:Single Frechet derivative}
this equation can be rewritten in terms of the Fréchet derivative
as 
\begin{equation}
\text{d}S\left(X_{Q\eta}\right)=\sum_{i=1}\int_{\Omega}\left[D_{L}\left(Q^{i}\right)\eta^{i}-\text{div}\left(A^{i}\right)\eta^{i}\right]\quad.
\end{equation}
Due to lemma \eqref{lem:The-symmetry-condition in Frechet style},
which related the symmetries and Fréchet derivatives, we can instead
write 
\begin{equation}
\text{d}S\left(X_{Q\eta}\right)=\int_{\Omega}\left[\sum_{i=1}D_{L}\left(Q^{i}\right)\left(\eta^{i}-\eta^{\alpha}\right)\right]-\int_{U}\text{div}\left(B\right)\eta^{\alpha}-\sum_{i=1}^{N}\int_{U}\text{div}\left(A^{i}\right)\eta^{i}\quad.
\end{equation}
By the definition of the non-local correction term $\Delta^{\alpha}\left(z\right)=\int_{\Omega}\sum_{i=1}^{N}D_{L}\left(Q^{i}\right)\left[\delta^{i}-\delta^{\alpha}\right]\left(z\right)$
we get 
\begin{equation}
\text{d}S\left(X_{Q\eta}\right)=\int_{U}\left\{ \Delta^{\alpha}\left(z\right)-\sum_{i=1}^{N}\text{div}_{\mathcal{D}^{i}}\left(A^{i}\right)\delta^{i}\left(z\right)-\int_{U}\text{div}_{M^{P}}\left(B\right)\delta^{\alpha}\left(z\right)\right\} \eta\left(z\right)\quad,\label{eq:Variation of the action in almost a symmetry}
\end{equation}
which is a generalised version of the local equation $\text{d}S=\int\,\text{div}\left(J_{\text{phys}}\right)\cdot\eta$.
Inserting this equation into Eq. \eqref{eq:General ward type identitiy},
we obtain 
\begin{equation}
\int_{U}\eta\left(z\right)\langle\text{Tr}\left(\CD{\FB}Q\right)\left(z\right)-\text{GCL}\left(z\right)+JQ\left(z\right)\rangle=0\quad,
\end{equation}
where GCL stays for the Generalised Conservation Law. Since the above
equation holds for all smooth functions $\eta$ on $U$. This leads
to the modified local Ward identities 
\begin{equation}
\langle\text{Tr}\left(\CD{\FB}Q\right)\left(z\right)-\text{GCL}\left(z\right)+JQ\left(z\right)\rangle=0\quad.
\end{equation}

\section{Conclusions}

In this paper we have proposed a geometrical treatment of symmetries
in non-local field theories, where the non-locality is due to a lack
of identification of field arguments in the action. We have shown
that the existence of a symmetry of the action leads to a generalised
conservation law, in which the usual conserved current acquires an
additional non-local correction term, and thus obtained a generalization
of the standard Noether theorem to such more exotic cases. We have
illustrated the general formalism by discussing a specific physical
example, in which this correction term can be interpreted as a dynamical
external mean potential coming from the structure of the non-local
field interaction.

Our analysis was focused on cases in which the geometrical Lagrangian
is a function on the first jet bundle but we believe that an extension
to higher jet bundles can be carried over from local theories without
any complications. On the other hand, it is not clear whether the
generalization to non-geometrical symmetries can be carried out in
the same way as it is done in the local case. The major difficulty
in this direction is the fact, that the total divergence on $\BM$
can, in general, change the equations of motion leading to non-equivalent
Lagrangians. A solution to this problem could lead to a full classification
of (local) symmetries for non-local field theories.

We hope that our analysis and results will bring new insights in different
areas of theoretical physics in which non-local field theories play
a role, including condensed matter physics and cosmology. However,
our main interest is the application of the above analysis to non-local
models of quantum gravity, especially in the framework of group field
theories (and thus, indirectly, of spin foam models and loop quantum
gravity). The detailed analysis of known symmetries in group field
theories and of their corresponding conservation laws is the subject
of a follow-up paper. However, the interest of such analysis can already
be envisaged. In particular, it will be important to develop the symmetry
analysis of such models \cite{josephNoether,GFTdiffeos} in a more
systematic way and to use it to shed more light into their quantum
geometric properties. The consequences of existing symmetries, encoded
in the generalised Ward identities \cite{josephWard,LOR,LOR2}, will
also play an important role in the analysis of group field theory
renormalization \cite{GFTrenorm,GFTrenorm2,GFTrenorm3,GFTrenorm4,GFTrenorm5,GFTrenorm6,GFTrenorm7,GFTrenorm8,GFTrenorm9,GFTrenorm10,GFTrenorm11}.
Finally, we expect symmetries and generalised conservation laws to
be crucial for the further development and physical analysis of the
effective cosmological dynamics extracted from group field theory
models of quantum gravity \cite{GFTcosmo}.

\section*{Acknowledgements}

We are very grateful to the members of the quantum gravity group at
the AEI, and especially to Joseph Ben Geloun for several discussions
and useful comments on this work.


\begin{thebibliography}{10}
\bibitem{Noether:1971aa} E. Noether, (1971) (translated in english
2005) arXiv:physics/0503066

\bibitem{Leggett1} A.~Leggett, ``Quantum Liquids,'' Oxford University
Press (2006); 

\bibitem{Leggett2}A.~J.~Leggett, Rev.\ Mod.\ Phys.\ \textbf{73},
307 (2001).

\bibitem{Kristensen:1952aa} P. Kristensen and C. Møller, Dan. Mat.
Fys. Medd. 27, no. 7 (1952)

\bibitem{Pauli:1953aa} W. Pauli, IL Nuovo Cimento, vol 10, no. 5,
pp 648-667 (1953)

\bibitem{nonlocal2} S. Deser, R. Woodard, Phys.Rev.Lett.99:111301,2007,
arXiv:0706.2151 {[}astro-ph{]}; 

\bibitem{nonlocal3}S. Foffa, M. Maggiore, E. Mitsou, Int.J.Mod.Phys.
A29 (2014) 1450116, arXiv:1311.3435 {[}hep-th{]}; 

\bibitem{nonlocal4}T. Koivisto, Phys.Rev. D77 (2008) 123513, arXiv:0803.3399
{[}gr-qc{]}

\bibitem{nonlocal5}S. Nojiri, S. D. Odintsov Phys.Lett. B659 (2008)
821-826, arXiv:0708.0924 {[}hep-th{]}

\bibitem{nonlocal6}S. Jhingan, S. Nojiri, S.D. Odintsov, M. Sami,
I. Thongkool, S. Zerbini Phys.Lett. B663 (2008) 424-428 arXiv:0803.2613
{[}hep-th{]}

\bibitem{nonlocCosmo} T. Banks, Nucl.Phys. B309 (1988) 493; 

\bibitem{nonlocCosmo2}T. Banks, Phys.Scripta T117 (2005) 56-63, hep-th/0310288; 

\bibitem{nonlocCosmo3}C. Prescod-Weinstein, L. Smolin, Phys.Rev.
D80 (2009) 063505, arXiv:0903.5303 {[}hep-th{]}

\bibitem{nonlocBH} S. Giddings, Phys.Rev. D74 (2006) 106005, hep-th/0605196;

\bibitem{nonlocBH2}S. Giddings, Phys.Rev. D74 (2006) 106006, hep-th/0604072;

\bibitem{nonlocBH3}P. Nicolini, arXiv:1202.2102 {[}hep-th{]}

\bibitem{sorkin} R. Sorkin, AIP Conf.Proc. 957 (2007) 142-153, arXiv:0710.1675
{[}gr-qc{]}

\bibitem{dsr} G. Amelino-Camelia, L. Freidel, J. Kowalski-Glikman,
L. Smolin, Phys.Rev. D84 (2011) 084010, arXiv:1101.0931 {[}hep-th{]}; 

\bibitem{dsr2}S. Hossenfelder, Phys.Rev.Lett. 104 (2010) 140402,
arXiv:1004.0418 {[}hep-ph{]}

\bibitem{NCG} G. Gubitosi, F. Mercati, Class.Quant.Grav. 30 (2013)
145002, arXiv:1106.5710 {[}gr-qc{]}; 

\bibitem{NCG2}S. Majid, Lect.Notes Phys. 541 (2000) 227-276, hep-th/0006166

\bibitem{emergence} B.~L.~Hu, Int.\ J.\ Theor.\ Phys.\ \textbf{44},
1785 (2005) {[}gr-qc/0503067{]}; 

\bibitem{emergence2}T. Konopka, F. Markopoulou, S. Severini, Phys.Rev.D77:104029,200,
arXiv:0801.086;

\bibitem{emergence3}D.~Oriti, PoS QG \textbf{-PH} (2007) 030 {[}arXiv:0710.3276
{[}gr-qc{]}{]}; 

\bibitem{emergence4}D.~Oriti, Stud.\ Hist.\ Philos.\ Mod.\ Phys.\ \textbf{46}
(2014) 186 {[}arXiv:1302.2849 {[}physics.hist-ph{]}{]}.

\bibitem{GFT} D.~Oriti, in \textsl{Mathematical and Physical Aspects
of Quantum Gravity,} B. Fauser, et al. (eds) (Birkhaeuser, Basel,
2006), {[}gr-qc/0512103{]};

\bibitem{GFT2}A.~Baratin and D.~Oriti, J.\ Phys.\ Conf.\ Ser.\ \textbf{360},
012002 (2012) {[}arXiv:1112.3270 {[}gr-qc{]}{]}; 

\bibitem{GFT3}T.~Krajewski, PoS QGQGS \textbf{2011}, 005 (2011)
{[}arXiv:1210.6257 {[}gr-qc{]}{]}; 

\bibitem{GFT4}D.~Oriti, in \textsl{The Planck Scale}, J. Kowalski-Glikman,
et al. (eds) AIP: conference proceedings (2009), arXiv:0912.2441 {[}hep-th{]}; 

\bibitem{GFT5}D.~Oriti, in: ``Foundations of space and time'',
G. Ellis, J. Marugan, A. Weltman (eds.), Cambridge University Press
(2012), arXiv:1110.5606 {[}hep-th{]}

\bibitem{GFT-LQG} D.~Oriti, arXiv:1310.7786 {[}gr-qc{]}; 

\bibitem{GFT-LQG2}D.~Oriti, in \lq Loop Quantum Gravity\rq, A.
Ashtekar, J. Pullin (eds), World Scientific (to appear), arXiv:1408.7112
{[}gr-qc{]}; 

\bibitem{GFT-LQG3}D.~Oriti, J.~P.~Ryan and J.~Thürigen, New J.Phys.
17 (2015) 2, 023042 arXiv:1409.3150 {[}gr-qc{]}.

\bibitem{LQG} A.~Ashtekar and J.~Lewandowski, Class.\ Quant.\ Grav.\ \textbf{21},
R53 (2004) {[}gr-qc/0404018{]}; 

\bibitem{LQG2}T.~Thiemann, ``Modern canonical quantum general relativity'',
Cambridge, UK: Cambridge Univ. Pr. (2007) 819p. C.~Rovelli, PoS QGQGS
\textbf{2011}, 003 (2011) {[}arXiv:1102.3660 {[}gr-qc{]}{]}.

\bibitem{tensor} R.~Gurau and J.~P.~Ryan, SIGMA \textbf{8} (2012)
020 {[}arXiv:1109.4812 {[}hep-th{]}{]}; 

\bibitem{tensor2}V. Rivasseau, Fortsch.Phys. 62 (2014) 81-107, arXiv:1311.1461
{[}hep-th{]}

\bibitem{martin} M. Bojowald, Rept.Prog.Phys. 78 (2015) 023901, arXiv:1501.04899
{[}gr-qc{]}

\bibitem{GFTcosmo} S.~Gielen, D.~Oriti and L.~Sindoni, Phys.\ Rev.\ Lett.\ \textbf{111}
(2013) 031301 {[}arXiv:1303.3576 {[}gr-qc{]}{]}; 

\bibitem{GFTcosmo2}S.~Gielen, D.~Oriti and L.~Sindoni, JHEP 1406
(2014) 013, arXiv:1311.1238 {[}gr-qc{]}; 

\bibitem{GFTcosmo3}S.~Gielen, Class.Quant.Grav. 31 (2014) 155009,
arXiv:1404.2944 {[}gr-qc{]}; G. Calcagni, {\em Phys.\ Rev.\ D}
\textbf{90} (2014) 064047, arXiv:1407.8166{[}gr-qc{]}; 

\bibitem{GFTcosmo4}S.~Gielen and D.~Oriti, New J.\ Phys.\ \textbf{16},
123004 (2014) {[}arXiv:1407.8167 {[}gr-qc{]}{]}; S.~Gielen, Phys.Rev.
D91 (2015) 4, 043526, arXiv:1411.1077 {[}gr-qc{]}; 

\bibitem{GFTcosmo5}L.~Sindoni, arXiv:1408.3095 {[}gr-qc{]}; 

\bibitem{GFTcosmo6}S. Gielen, arXiv:1505.07479 {[}gr-qc{]}; 

\bibitem{GFTcosmo7}D. Oriti, D. Pranzetti, J. Ryan, L. Sindoni, arXiv:1501.00936
{[}gr-qc{]}

\bibitem{Bloch:1950aa} C. Bloch, Mat. Fys. Medd. XXVI, nr. 1 (1950)

\bibitem{Marnelius:1973aa} R. Marnelius, Phys. Rev. D 8, 2472 (1973)

\bibitem{Huang:2012aa} Z. Huang, (2012), arXiv:1203.1149 {[}math-ph{]}

\bibitem{Yukawa:1950aa} H. Yukawa, Phys. Rev. 77, 219 (1950)

\bibitem{josephNoether} J. Ben Geloun, J.Math.Phys. 53 (2012) 022901,
arXiv:1107.3122 {[}hep-th{]}

\bibitem{Abraham:2008aa} R. Abraham and J. E. Marsden, ``Foundations
of Mechanics'' (second edition), AMS Chelsea (2008)

\bibitem{Olver:1998aa} P. J. Olver, ``Applications of Lie Groups
to Differential Equations'' (second edition), Springer (1998)

\bibitem{Giaquinta:1996aa}M. Giaquinta and S. Hildebrandt, Calculus
of Vatiations I, Springer, (1996)

\bibitem{Kurz:2011aa} S. Kurz, (2011) arXiv:1106.0926v1

\bibitem{Nair:2004aa} V. P. Nair, ``Quantum Field Theory A Modern
Perspective'', Springer (2004)

\bibitem{GFTdiffeos} A. Baratin, F. Girelli, D. Oriti, Phys.Rev.
D83 (2011) 104051, arXiv:1101.0590 {[}hep-th{]}

\bibitem{josephWard} J. Ben Geloun, J.Phys. A44 (2011) 415402, arXiv:1106.1847
{[}hep-th{]}

\bibitem{LOR} D.~O.~Samary, C.~I.~Pérez-Sánchez, F.~Vignes-Tourneret
and R.~Wulkenhaar, arXiv:1411.7213 {[}hep-th{]}; 

\bibitem{LOR2}V. Lahoche, D. Oriti, V. Rivasseau, JHEP 1504 (2015)
095, arXiv:1501.02086 {[}hep-th{]}

\bibitem{GFTrenorm} J.~Ben Geloun and V.~Rivasseau, Commun.\ Math.\ Phys.\ \textbf{318},
69 (2013) {[}arXiv:1111.4997 {[}hep-th{]}{]}; 

\bibitem{GFTrenorm2}S.~Carrozza, D.~Oriti and V.~Rivasseau, Commun.\ Math.\ Phys.\ \textbf{327},
603 (2014) {[}arXiv:1207.6734 {[}hep-th{]}{]}; 

\bibitem{GFTrenorm3}D.~O.~Samary and F.~Vignes-Tourneret, Commun.\ Math.\ Phys.\ \textbf{329},
545 (2014) {[}arXiv:1211.2618 {[}hep-th{]}{]}; 

\bibitem{GFTrenorm4}S.~Carrozza, D.~Oriti and V.~Rivasseau, Commun.\ Math.\ Phys.\ \textbf{330},
581 (2014) {[}arXiv:1303.6772 {[}hep-th{]}{]}; 

\bibitem{GFTrenorm5}J.~Ben Geloun, Commun. Math. Phys. \textbf{332},
117--188 (2014) {[}arXiv:1306.1201 {[}hep-th{]}{]}; 

\bibitem{GFTrenorm6}J.~Ben Geloun, ``On the finite amplitudes for
open graphs in Abelian dynamical colored Boulatov-Ooguri models,''
J.\ Phys.\ A \textbf{46}, 402002 (2013) {[}arXiv:1307.8299 {[}hep-th{]}{]}; 

\bibitem{GFTrenorm7}S.~Carrozza, Springer Theses, 2014 (Springer,
NY, 2014), arXiv:1310.3736 {[}hep-th{]}; 

\bibitem{GFTrenorm8}J.~Ben Geloun, Class.\ Quant.\ Grav.\ \textbf{29},
235011 (2012) {[}arXiv:1205.5513 {[}hep-th{]}{]}; 

\bibitem{GFTrenorm9}D.~O.~Samary, ``Beta functions of $U(1)^{d}$
gauge invariant just renormalizable tensor models,'' Phys.\ Rev.\ D
\textbf{88}, 105003 (2013) {[}arXiv:1303.7256 {[}hep-th{]}{]}; 

\bibitem{GFTrenorm10}S.~Carrozza, ``Discrete Renormalization Group
for SU(2) Tensorial Group Field Theory,'' Ann. Inst. Henri Poincaré
Comb. Phys. Interact. 2 (2015), 49-112, arXiv:1407.4615 {[}hep-th{]}; 

\bibitem{GFTrenorm11}S. Carrozza, Phys. Rev. D 91, 065023 (2015),
arXiv:1411.5385 {[}hep-th{]}\end{thebibliography}
\end{document}